\documentclass[11pt]{article}
\usepackage[margin=1in]{geometry}

\usepackage{amssymb,amsmath,amsthm,amsopn}
\usepackage{mathrsfs}

\usepackage[usenames,dvipsnames]{color}
\usepackage{xcolor}
\usepackage[colorlinks]{hyperref}

\usepackage{tikz}
\usepackage{pgfplots}
\usepackage{graphicx}
\usepackage{mathtools}
\usepackage{cite}

\usepackage{thmtools, thm-restate}
\usepackage{xparse}
\usepackage{everysel}

\allowdisplaybreaks

\newcommand{\defn}[1]{\textcolor{blue}{#1}}

\newtheorem{lemma}{Lemma}[section]
\newtheorem{theorem}[lemma]{Theorem}
\newtheorem{proposition}[lemma]{Proposition}
\newtheorem{corollary}[lemma]{Corollary}
\numberwithin{equation}{section}

\declaretheoremstyle[
    headfont=\bfseries,
    notefont=\itshape,
    bodyfont=\normalfont,
    headpunct={},
    headformat={\defn{ \NAME\ \NUMBER}},
]{definitionsty}

\EverySelectfont{%
\fontdimen2\font=0.4em
\fontdimen3\font=0.2em
\fontdimen4\font=0.1em
\fontdimen7\font=0.1em
\hyphenchar\font=`\-
}

\declaretheoremstyle[
    headfont=\ttfamily\bfseries,
    notefont=\itshape,
    bodyfont=\ttfamily,
    headpunct={\textcolor{red}{.}},
    headformat={\textcolor{red}{ \NAME}},
]{specialremarksty}

\theoremstyle{definition}
\newtheorem*{remark}{Remark}
\newtheorem*{remarks}{Remarks}

\theoremstyle{specialremarksty}

\newcommand{\numset}[1]{\mathbb{#1}}
	\newcommand{\C}{\numset{C}}	
	\newcommand{\R}{\numset{R}}
	
	\newcommand{\Z}{\numset{Z}}
	\newcommand{\N}{\numset{N}}
	\def\cc{\numset{C}}

	\def\nn{\numset{N}}
	\def\infspec{\inf\mathrm{sp}\,}
	
\newcommand{\one}{\mathrm{Id}}
\def\id{\mathrm{Id}}
\def\i{{\rm i}}
\renewcommand{\d}{\operatorname{d}\!}
\newcommand{\inv}{^{-1}}

	\DeclareMathOperator{\spr}{spr}
	\renewcommand{\sp}{\operatorname{sp}}

	\newcommand{\tr}{\operatorname{Tr}}
	\DeclareMathOperator{\ran}{Ran}

\newcommand{\spacefont}[1]{\mathcal{#1}}
	\renewcommand{\H}{\spacefont{H}}
	\newcommand{\B}{\spacefont{B}}
	\newcommand{\I}{\spacefont{I}}

\newcommand{\sys}{\mathcal{S}}
\newcommand{\env}{\mathcal{E}}
\renewcommand{\L}{\mathcal{L}}
	\newcommand{\kat}[1]{W_{#1}}
	\newcommand{\Kat}[1]{K_{#1}}
	\newcommand{\Phase}[1]{\Phi_{#1}}

\newcommand{\invar}{^\textnormal{inv}}
\newcommand{\init}{^\textnormal{i}}
\newcommand{\fin}{^\textnormal{f}}
\DeclareDocumentCommand \envstate { o } {%
  \IfNoValueTF {#1} {%
    \xi%
  }{%
    \xi_{#1}%
  }%
}
\DeclareDocumentCommand \sysstate { o } {%
  \IfNoValueTF {#1} {%
    \rho%
  }{%
    \rho_{#1}%
  }%
}

 \newcommand{\Udiff}{X}

\newcommand{\ham}{h}
	\newcommand{\hsys}{\ham_{\sys}}
	\DeclareDocumentCommand \henv { o } {%
	  \IfNoValueTF {#1} {%
	    \ham_{\mathcal{E}}%
	  }{%
	    \ham_{\mathcal{E}_{#1}}%
	  }%
	}

\begin{document}

\author{ Eric Hanson$^{1}$, Alain Joye$^{2}$, Yan Pautrat$^{3}$, Renaud Raqu\'epas$^{1}$
}

\title{Landauer's Principle in Repeated Interaction Systems}

\maketitle

\thispagestyle{empty}

\begin{center}\small
$^1$ Department of Mathematics and Statistics, \\
McGill University\\
1005-805 rue Sherbrooke Ouest \\
Montr\'eal~QC, H3A 0B9, Canada

\medskip
$^2$ Universit\'e Grenoble Alpes \\ CNRS, Institut Fourier \\ F-38000 Grenoble, France

\medskip
$^3$ Laboratoire de Math\'ematiques d'Orsay,\\
Univ. Paris-Sud, CNRS, Universit\'e
Paris-Saclay\\  91405 Orsay, France
\end{center}

\begin{abstract}
We study Landauer's Principle for Repeated Interaction Systems (RIS) consisting of a reference quantum system $\mathcal{S}$ in contact with a structured environment $\mathcal{E}$ made of a chain of independent quantum probes; $\mathcal{S}$ interacts with each probe, for a fixed duration, in sequence. We first adapt  Landauer's lower bound, which relates the energy variation of the environment $\mathcal{E}$  to a decrease of entropy of the system $\mathcal{S}$  during the evolution, to the peculiar discrete time dynamics of RIS. Then we consider RIS with a structured environment $\mathcal{E}$ displaying small variations of order $T^{-1}$ between the successive probes encountered by $\mathcal{S}$, after $n\simeq T$ interactions, in keeping with adiabatic scaling. We establish a discrete time non-unitary adiabatic theorem to approximate the reduced dynamics of $\mathcal{S}$ in this regime, in order to tackle the adiabatic limit of Landauer's bound. We find that saturation of Landauer's bound is related to a detailed balance condition on the repeated interaction system, reflecting the non-equilibrium nature of the repeated interaction system dynamics.
This is to be contrasted with the generic saturation of Landauer's bound known to hold for continuous time evolution of an open quantum system interacting with a single thermal reservoir in the adiabatic regime. 

\end{abstract}

\section{Introduction} \label{sec:Introduction}

Landauer's Principle deals with the variations of thermodynamical quantities of a reference system $\sys$ interacting with an environment $\env$ at \emph{inverse temperature} $\beta$.\footnote{The inverse temperature is $\beta := (k_B T)^{-1}$, where $k_B$ is Boltzmann's constant and $T$ is the temperature. We will set $k_B=1$ so that the environment is at temperature $T = \beta^{-1}$. We reserve the symbol $T$ for adiabatic scaling.} It can be stated as a general lower bound on the variation of energy $\Delta Q_\env$ of the environment necessary for the system~$\sys$ to undergo some entropy decrease $\Delta S_\sys$ under the unitary evolution of the coupled system $\sys+\env$, in terms of this entropy variation. As we recall in Section\nobreakspace \ref {sec:Landauer}, this bound is a direct consequence of the fact that the difference between the energy and entropy variations considered is a non-negative relative entropy, or entropy production, $\sigma$: one has $\Delta S_\sys+\sigma=\beta \Delta Q_\env$, \cite{La,RW13}. See \cite{JP14} for generalisations and a thorough discussion of this principle. 
A typical application of this principle arises in state engineering processes, when starting from some initial state $\sysstate\init$ of $\sys$, one aims for a given  target state $\sysstate\fin$ for $\sys$, by a suitable, possibly time dependent, coupling with an environment.
While Landauer's lower bound is known not to be optimal, a natural question in this framework is to ask under which circumstances or in which regimes it can be saturated, i.e. when the entropy production $\sigma\simeq 0$. The derivation of Landauer's Principle and the laws of thermodynamics suggest that for quasi static processes saturation should occur. Indeed, the case where the Hamiltonian of the coupled system $\sys+\env$ is time-dependent and slowly varying is one such regime. This regime describes, for example, the adiabatic switching of the interaction between the reference system and the environment. That Landauer's bound is saturated in the adiabatic limit is proven in great generality in~\cite{JP14}. The authors show that the time evolution of the entropy production $\sigma$ vanishes in the adiabatic limit, by means of an adiabatic theorem suitable for generators without gap in their spectrum,~\cite{AE99}.

In this paper, we reconsider Landauer's Principle and some of its properties for quantum dynamical systems called Repeated Interaction Systems (RIS), and described in Section\nobreakspace \ref {sec:RIS}. They consist of a reference system~$\sys$ interacting with a structured environment $\env$ defined as a chain of independent quantum systems, the probes, which the reference system~$\sys$ interacts with, in sequence, for a fixed duration. After~$\sys$ and the $k$-th probe interact by means of a unitary operator, the probe is discarded, together with the energy it has exchanged with $\sys$. The fact that a used probe never interacts with $\sys$ again provides the chain of probes the status of an effective environment. 
When the probes are all identical and, say, in a thermal state as they start interacting with the system~$\sys$, the latter is driven for large times to a non equilibrium steady state, the characteristics of which depend on the state of the probes, not on the initial state $\sysstate\init$ of $\sys$. See  \cite{BJM06,BJM08}, and the review \cite{BJM14}. 
The physical archetype of RIS is the one atom maser. The reference system~$\sys$ in this case is the electro-magnetic field in a laser cavity which interacts with a beam of atoms entering the cavity one by one. The sequence of atoms arriving in the cavity in a thermal state provides a structured environment $\env$ for the reference system $\sys$. 

As the fate of the probes after they interact with the reference system is irrelevant for our purpose, we focus on the effective dynamics of the state of $\sys$ given by tracing out the probes degrees of freedom after each unitary interaction step. Thanks to the nature of the dynamics of RIS,  the reduced dynamics on $\sys$ is given by a linear map ${\L}_k$, that describes the evolution of the state of $\sys$ after interaction with the $k$-th probe, and encodes the properties of the state of this probe; see \cite{BJM14}.
Each map ${\L}_k$ is completely positive and trace preserving (CPTP), possesses an invariant state at least (by e.g. Brouwer's theorem) and, as an operator acting on the set of states on $\sys$ endowed with the trace norm, is a contraction; see, e.g. \cite{wolftour}. Therefore, generalising the setup to include variable interactions between $\sys$ and the probes as well as variable incoming probe states, an initial state~$\sysstate\init$ of the reference system undergoes a discrete time evolution to reach $\sysstate\fin={\L}_n {\L}_{n-1} \cdots {\L}_1\sysstate\init$ after interacting with the first $n$ probes of the environment. Landauer's Principle for RIS is then formulated in a similar way as for continuous quantum systems in Section\nobreakspace \ref {ssec:LanRIS}: after interacting with the $k$-th probe, the difference between the energy variation of that $k$-th probe and the entropy variation of the reference system is given by a non-negative relative entropy, or entropy production, $\sigma_k$, depending on $k$. Summing over all involved probes, we derive a version of Landauer's lower bound on the energy variation of the structured environment $\env$ in terms of the entropy variation between  $\sysstate\init$  and $\sysstate\fin={\L}_n {\L}_{n-1} \cdots {\L}_1\sysstate\init$ for RIS.

Our next goal is to study Landauer's bound in case the RIS undergoes quasi static transformations, in the spirit of \cite{JP14}. We consider environments $\env$ consisting of probes whose states may slightly vary between successive elements, and interactions between the system $\sys$ and the successive probes which may slightly vary from one probe to the other. The sizes of these slight variations are of order $T^{-1}\ll 1$, and thus give rise to $T$-dependent CPTP contractive maps ${\L}_{k,T}$, $k\in \N$, whose successive differences are of order $T^{-1}$, where $T$ is the maximum number of steps considered and the \emph{adiabatic parameter} of the process.  We then consider the evolution of the coupled system $\sys$+$\env$ after $\sys$ has interacted with $n\simeq T$ probes, in keeping with the adiabatic scaling. To analyse the saturation of the RIS Landauer bound for $T\gg 1$, we need to assess the large $T$ behaviour of $\sum_{k=1}^{T}\sigma_{k, T}$, where~$\sigma_{k, T}$ denotes the entropy production at the $k$-th step of the evolution. 
To do this, we formulate and prove an adiabatic theorem for contractive non-unitary discrete time evolutions in Section\nobreakspace \ref {sec:Adiabatic}, allowing us to get the asymptotics of the $T$-dependent states $\sysstate\fin_{k, T}={\L}_{k,T} {\L}_{k-1, T} \cdots {\L}_{1,T}\sysstate\init$ for all~$k \leq  T$, in the limit $T\rightarrow \infty$. 

This technical result, a mathematical result in its own right, extends the range of situations in which adiabatic approximations for linear evolution operators are available. We note that adiabatic approximations for unitary evolution operators generated by slowly varying time dependent self-adjoint generators with gaps in their spectrum can be found in \cite{KatoPaper,Nen,ASY} and, without gap assumptions, in \cite{AE99,Teufel}. Extensions to non-unitary semigroups of contractions with or without gap condition can be found in \cite{ASF,Joye,AFGG,Schmid}. On the other hand, discrete time adiabatic theorems have been proven in \cite{DKS98,Tan11} for unitary groups only. Our contribution thus provides an adiabatic approximation in the discrete time non-unitary case.

As a consequence of this result, we obtain the asymptotics of the relative entropy $\sigma_{k, T}$  as~$T\rightarrow \infty$. As can be expected, for suitable initial states $\sysstate\init$, the instantaneous invariant states of the CPTP maps ${\L}_{k,T}$, $k\in \N$, provide the leading term of $\sysstate\fin_{k, T}$ in the adiabatic scaling. We establish in Section\nobreakspace \ref {sec:Relative-entropy} an efficient perturbation theory of the relative entropy of two states that are close to a given fixed state (the invariant state of ${\L}_{k,T}$ in our case) and this allows us to use the leading term of $\sysstate\fin_{k, T}$ to compute the non-zero leading order of the relative entropy $\sigma_{k, T}$  for $T$ large. In turn, 
this makes it possible to analyse the total entropy production $\sum_{k=1}^{T}\sigma_{k, T}$ as $T\rightarrow \infty$ as a function of exterior parameters, in particular, the coupling strength between the probes and the reference system. This analysis identifies a key quantity $X_{k,T}$ which controls the vanishing of the entropy production; see Corollary\nobreakspace \ref {cor:sup_inf_special_RIS_Xs}. The case $X_{k,T} = 0$ is shown to relate to a detailed balance condition in Lemma\nobreakspace \ref {lemma_caracX}, and implies vanishing of the total entropy production. We have an explicit example demonstrating $X_{k,T} = 0$; another, $\|X_{k,T}\|> 0$, which yields divergent total entropy production (with an explicit rate).  The second case is expected generically, due to the nature of the RIS dynamics which imposes a change of probe at each time step (see Section\nobreakspace \ref {sec_concludingremarks} for a more detailed discussion).

\paragraph{Acknowledgements.} The research of Y.P. was supported by ANR contract ANR-14-CE25-0003-0. Y.P. also wishes to thank UMI-CRM for financial support, and McGill University for its hospitality. The research of E.H. and R.R. was partly supported by ANR contracts ANR-13-BS01-0007 and ANR-14-CE25-0003-0. The research of R.R. was also partly supported by NSERC. E.H. and R.R. wish to thank the Institut Fourier, where part of this research was carried out, for its support and hospitality. We would like to thank V. Jak\v{s}i\'{c} for informative discussions and suggestions about this project, and S. Andr\'eys for discussions that led to the formulation of Lemma \ref{lemma_caracX}.

\section{Landauer's Principle} \label{sec:Landauer}

We will state and derive Landauer's Principle in the simple case where both the small system $\sys$ and the reservoir $\env$ are described by finite dimensional Hilbert spaces (such a reservoir is called \textit{confined}). We also assume that the total system $\sys+\env$ is closed, and its evolution is therefore given by a unitary operator. This derivation is given in \cite{RW13} and extended in \cite{JP14}.

\subsection{Definitions and statement of Landauer's Principle}

The system of interest $\sys$ is described by a finite dimensional Hilbert space $\H_\sys$ with self-adjoint Hamiltonian $\hsys$. The environment $\env$ is described by a finite dimensional Hilbert space $\H_\env$ with Hamiltonian $\henv$. In this finite dimensional framework, the physical state of $\sys$ (respectively $\env$, $\sys+\env$) is described by a non-negative trace-class operator on $\H_\sys$ (resp. $\H_\env$, $\H_\sys \otimes \H_\env$) with trace equal to one. 
We say that the state is faithful if it is positive definite.

Let $\sysstate\init$ be the initial state of the system. We assume that the environment is initially at thermal equilibrium at inverse temperature $\beta$, that is in the Gibbs state $\envstate\init = \exp[-\beta \henv]/Z$, where $Z = \tr(\exp[-\beta \henv])$. 

The two parts $\sys$ and $\env$ are initially uncoupled and the initial state of the joint system is therefore $\sysstate\init \otimes \envstate\init$. The evolution of the full system is described by a unitary $U \in \B(\H_\sys \otimes \H_\env)$, bringing the state to $U \sysstate\init \otimes \envstate\init U^*$. The two systems are then decoupled to obtain final states 
\[\sysstate\fin = \tr_\env(U \sysstate\init \otimes \envstate\init U^*), \qquad \envstate\fin = \tr_\sys(U \sysstate\init \otimes \envstate\init U^*)\]
(the \textit{partial traces} $ \tr_{\env}$ and $\tr_\sys$ are defined by $\tr_{\env}(A\otimes B):= \tr(B) \, A$, $\tr_{\sys}(A\otimes B):= \tr(A) \, B$ when $A\otimes B$ is an operator on~$\H_\sys\otimes \H_\env$).

We define the decrease of entropy of the system and the increase of energy of the environment
\[
	\Delta S_\sys := S(\sysstate\init) - S(\sysstate\fin), \qquad
	\Delta Q_\env := \tr(\henv\envstate\fin) - \tr(\henv\envstate\init),
\]
where $S(\eta)$ is the von Neumann entropy $S(\eta):= -
\tr(\eta \log\eta)$. 
We will drop the $\env$ and $\sys$ subscripts of $\Delta Q_\env$ and $\Delta S_\sys$ for notational simplicity.
\smallskip

Let the relative entropy of two faithful states $\eta$, $\nu$ be given by $S(\eta|\nu):= 
\tr\big(\eta (\log \eta - \log \nu)\big)$. Then $S(\eta|\nu)\geq 0$, with equality if and only if $\eta=\nu$ (see Section 2.6 in \cite{JOPP12}). Let
\begin{equation}
\sigma = S(U \sysstate\init \otimes \envstate\init \,U^* | \sysstate\fin \otimes \envstate\init). \label{eq:production}
\end{equation}
A straightforward computation gives
\begin{align}	
\sigma&= - S(U \sysstate\init\otimes \envstate\init U^*) -
\tr\big( U \sysstate\init\otimes \envstate\init U^* \, (\log \sysstate\fin \otimes \one)\big) - 
\tr\big( U \sysstate\init\otimes \envstate\init U^* \,  ( \one \otimes \log \envstate\init)\big) \notag \\
&= -S(\sysstate\init\otimes \envstate\init ) + S(\sysstate\fin) -
\tr(\envstate\fin \log \envstate\init) \notag \\
&= -S(\sysstate\init) - S(\envstate\init ) + S(\sysstate\fin)  -
\tr(\envstate\fin \log \envstate\init) \notag \\
&=-\Delta S +  
\beta\Delta Q. \label{eq:balance}
\end{align}
This \textit{entropy balance equation} implies the \emph{Landauer bound}
\begin{align}
	\Delta Q \geq  
	\beta^{-1}\Delta S. \label{eq:Landauer}
\end{align}
by nonnegativity of relative entropies. 
More careful examination shows that equality holds in~\eqref{eq:Landauer} if and only if $\Delta S = \Delta Q = 0$ (see \cite{JP14}), in which case $\xi\fin=\xi\init$, and $\rho\init$, $\rho\fin$ are unitarily equivalent. Therefore, the saturation of inequality~\eqref{eq:Landauer} in the finite time setting holds only in trivial cases. 

Landauer's Principle in the context of more general $C^*$-dynamical systems, allowing the treatment of infinite dimensional environment Hilbert space $\H_\env$, is discussed in \cite{JP14}. The dynamics is still assumed to be conservative (and therefore given by an automorphism). For repeated interaction systems, however, a description by conservative dynamics is impractical (see the discussion at the end of Section\nobreakspace \ref {ssec:LanRIS}). We therefore revert to non-unitary dynamics, which cannot be treated by the results of~\cite{JP14}.

\subsection{The adiabatic limit} \label{subsec_adiabatic}

In this subsection, we recall that the saturation  of inequality \eqref{eq:Landauer} also holds in the infinite time regime, in the so-called \textit{adiabatic limit}. We start by recalling this limit.

For a system with state space $\H$ and whose time evolution for $s \in [0,1]$ is described by the Schr\"odinger equation with time-dependent Hamiltonian $\ham(s) \in \B(\H)$,
\[ i\frac{\d}{\d s} U(s) = \ham(s)U(s), \ s \in [0,1], \mbox{ with }U(0) = \one,\] 
the \emph{adiabatic limit} concerns the solution $U_T(s)$ of the rescaled Schr\"odinger equation
\begin{equation} \label{eq:rescaled} i\frac{\d}{\d t} U_T(t)= \ham(t/T)\,U_T(t), \ t \in [0,T], \mbox{ with }U_T(0) = \one,\end{equation}
in the limit $T \to \infty$. In this context, $T$ represents the physical time scale over which the process takes place and taking the limit $T\to\infty$ corresponds to the process being ``infinitely slow'', or quasi static.

Since the seminal work of Born and Fock \cite{BornFock}, a variety of adiabatic theorems have been formulated and proven, of which Kato's remains one of the most representative. In his~1950 paper \cite{KatoPaper}, Kato considers for $s\in[0,1]$ an eigenvalue parametrized by a continuous function~$e_1(s)$ of $\ham(s)$, with twice continuously differentiable spectral projector $P_1(s)$, and assumes that~$e_1(s)$ is separated from the rest of the spectrum by a gap. He then constructs a family $W(t)$, $t\in[0,T]$, of unitary operators satisfying
\begin{equation} W(t)P_1(0) = P_1(t/T)\,W(t)\label{eq:intertwine}\end{equation}
\begin{equation} \label{eq:KatoAdiabaticEstimate}
\Big( U_T(t) - \exp\big(-\i T \int_{0}^{t/T} \!\! e(s) \d s \big)\, W(t)\Big)\,  P_1(0) = O(T^{-1}), \ \mbox{uniformly in }t\in [0,T].
\end{equation}
In particular, 
  \eqref{eq:KatoAdiabaticEstimate}  shows that, in the adiabatic limit, the unitary dynamics $U_T$ essentially maps continuously every spectral subspace of $h(0)$ to the corresponding subspace of $h(s)$.

Fix $T>0$. Then, considering the unitary $U_T=U_T(1)$, the preceding subsection relates the quantities $\Delta S_T$ and $\Delta Q_T$ by
$ 	
\Delta S_T + \sigma_T = \beta \Delta Q_T 
$
with 
$
\sigma_{T} =  S\big(U_T (\sysstate\init \otimes \envstate\init) U_T^*| \sysstate_T\fin \otimes \envstate\init\big).
$

Another adiabatic theorem, that of Avron--Elgart (which does not require the eigenvalue~$e_1(s)$ to be isolated in the spectrum of $h(s)$~\cite{AE99}), is used in \cite{JP14}, together with Araki's perturbation theory of KMS states, to prove that (under ergodic assumptions on the dynamics of the system corresponding to different values of $t$), the map $\eta\mapsto U_T \eta U_T^*$ satisfies a relation similar to \eqref{eq:KatoAdiabaticEstimate} and maps to order $o(1)$ the state $\rho\init\otimes\xi\init$ to $\rho_T\fin\otimes\xi\init$. This implies 
$\lim_{T \to \infty} \sigma_{T}= 0$
and shows that Landauer's inequality is typically saturated in the adiabatic limit, at least in systems undergoing a (time-dependent) Hamiltonian evolution.

\section{Repeated interaction systems} \label{sec:RIS}

 Repeated interaction systems (RIS) form a special class of open quantum systems introduced in various forms (see \cite{AttalPautrat, Kummerer, BJM06, BruPil} and the review \cite{BJM14}; also see \cite{OQS1,OQS2,OQS3} for a discussion of the Hamiltonian and Markovian approaches). They consist of an open quantum system where the environment is a chain of outer systems called \textit{probes}, and where the small system $\sys$ interacts sequentially with each probe, one at a time and e.g. for a duration $\tau$, and the degrees of freedom of the probe are traced out before a new probe is brought in. In a RIS, the dynamics of the system are Hamiltonian during the interaction with a fixed probe; they are Markovian in that the environment degrees of freedom are traced out after a probe has been used up, so that (if all probes are identical) the evolution at times $(k\tau)_{k}$ forms a semigroup. In addition, the generators of the Markovian evolution can be directly expressed in terms of the Hamiltonian evolution in the corresponding time interval.

An important example of RIS is the \emph{one-atom maser}, in which the system $\sys$ consists of modes of an electromagnetic field inside a cavity and the probes $\env_k$ are atoms from a beam that interact with the field as they pass through the cavity. We will present two other examples in the next section.

\subsection{Mathematical description} \label{ssec:RISmathdesc}

To model the RIS, we describe the system $\sys$ by a finite-dimensional Hilbert space $\H_\sys$ with internal dynamics generated by a self-adjoint Hamiltonian  $\hsys \in \B(\H_\sys)$. Likewise, each probe is described by a finite-dimensional Hilbert space $\H_{\env\!,k}$, and a self-adjoint Hamiltonian $\henv[k] \in \B(\H_{\env\!,k})$. We assume that all of the probes' Hilbert spaces are identical: $\H_{\env\!,k} \equiv \H_\env$. We specify the initial state of the $k$-th probe, $\envstate[k]\init$, to be the Gibbs state at inverse temperature~$\beta_k$, that is
\begin{equation}\label{eq:thermal}
\envstate[k]\init := \frac{\exp-\beta_k \henv[k]}{\tr(\exp-\beta_k  \henv[k])}.
\end{equation}
	
We note that this is not the most general description of repeated interaction systems. For example, the Hilbert space of the ``small'' system in the one-atom maser setup is $\H_\sys \cong \Gamma_+(\C)$, the bosonic Fock space over $\C$, which is not finite dimensional. However, it is reasonable to approximate the system by restricting the description of the electromagnetic field to a finite number of energy levels (see \cite{BruPil} for a full treatment of the one-atom maser).
\smallskip
	
The evolution of the state of the system $\sys$ can be described in the following way: it starts from an initial state $\sysstate\init$, and assuming that it has evolved to $\sysstate[k-1]$ after interacting with the first~$k-1$ probes, the system interacts with the $k$-th probe as pictured in Figure\nobreakspace \ref {fig:RIS}.

The system and $k$-th chain element, which is initially in the state $\envstate[k]\init$, evolve for a  time $\tau_k$ via a potential $v_k$ with coupling constant $\lambda_k$, according to the unitary operator
\begin{equation} \label{eq_defUk}
U_k := \exp\big(-i\tau_k(\hsys \otimes \one + \one \otimes \henv[k] + \lambda_k v_k)\big),
\end{equation}
that is,
$\sysstate[k-1] \otimes \envstate[k]\init$ evolves to $U_k (\sysstate[k-1] \otimes \envstate[k]\init) U_k^*$.

Then, we trace out the chain element to obtain the system state
$$
	\sysstate[k] = \tr_{\env}\! \big(U_k (\sysstate[k-1] \otimes \envstate[k]\init) U_k^*\big).
$$
	
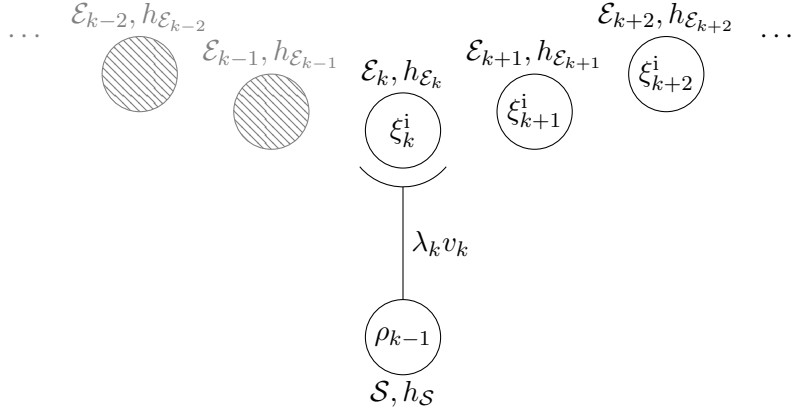
\begin{figure}[ht]
\begin{center}
\usetikzlibrary{patterns}
\begin{tikzpicture}[scale=0.5]
\draw  (0,0) ellipse (1 and 1);
\node at (0,0) {$\rho_{k-1}$};
\node at (0,-1.5) {$\mathcal{S},h_\sys$};

\draw (0,4) -- (0,1);
\node at (1,2.5) {$\lambda_k v_k$};

\draw  (0,5.5) ellipse (1 and 1);
\node at (0,7) {$\mathcal{E}_{k}, h_{\env_k}$};
\node at (0,5.5) {$\envstate[k]\init$};

\draw[gray,pattern=north west lines, pattern color=gray]  (-3.5,6) ellipse (1 and 1);
\node at (-3.5,7.5) {$\color{gray}\mathcal{E}_{k-1}, h_{\env_{k-1}}$};

\draw[gray,pattern=north west lines, pattern color=gray]  (-7,7) ellipse (1 and 1);
\node at (-7,8.5) {$\color{gray}\mathcal{E}_{k-2}, h_{\env_{k-2}}$};

\draw  (7,7) ellipse (1 and 1);
\node at (7,8.5) {$\mathcal{E}_{k+2}, h_{\env_{k+2}}$};
\node at (7,7) {${\envstate[k+2]\init}$};

\draw  (3.5,6) ellipse (1 and 1);
\node at (3.5,7.5) {$\mathcal{E}_{k+1}, h_{\env_{k+1}}$};
\node at (3.5,6) {${\envstate[k+1]\init}$};

\node at (-10,8) {$\color{gray}\cdots$};
\node at (10,8) {$\cdots$};

\draw (-1.1491,4.5358) arc (-140.0003:-40:1.5);
\end{tikzpicture}
	\caption{Schematic representation of a repeated interaction system at the beginning of the $k$-th step, that is at time $\sum_{n=1}^{k-1} \tau_n$. }
	\label{fig:RIS}
\end{center}
\end{figure}
This procedure defines a family of maps on $\I_1(\H_\sys)$, the trace-class operators on $\H_\sys$,
\begin{equation} \label{eq_defLk}
\begin{array}{cccc}
\L_k : & \I_1(\H_\sys) &\to &\I_1(\H_\sys) \\
& \eta &\mapsto &\tr_{\env} \big(U_k (\eta \otimes \envstate[k]\init) U_k^*\big)
\end{array}
\end{equation}
called the \emph{reduced dynamics}. In this language, the state of the system after step $k$ is therefore given by \[\sysstate[k] = \L_k \L_{k-1}\dotsm \L_1 \sysstate\init.\]
We therefore recover a Markovian form for the sequence of states $(\sysstate[k])_k$, and our definition \eqref{eq_defLk} shows that the expression for $\L_k$ can be derived in terms of the full Hamiltonian described by~$\hsys$,~$\henv[k]$, $v_k$, $\lambda_k$, $\tau_k$, and of the initial state $\xi_k\init$ of the probe.
\begin{remark}
We assume without loss of generality that $\lambda_k=\lambda$ for all $k\in \N$. The coupling constant~$\lambda$ will play a distinguished role below. 
\end{remark}
\smallskip

Let us recall the properties of the maps $\L_k$: each $\L_k$ is trace-preserving (i.e. $\tr \L_k(\eta)=\tr \eta$ for all $\eta$) and completely positive (i.e. for any $n\in \nn$, the map $\L_k \otimes \id$ on $\mathcal I_1(\H_\sys)\otimes \mathcal B(\cc^n)$ maps nonnegative operators to nonnegative operators, see \cite{OQS2} for more details). It is therefore called a \textit{CPTP map}, or a \textit{quantum channel}, on the ideal $\I_1(\H_\sys)$ of trace-class operators on $\H_\sys$ equipped with the trace-norm $\|\eta\|_1=\tr\big((\eta^* \eta)^{1/2}\big)$. Denote by $\|\L_k\|$ its uniform norm as an operator on $\I_1(\H_\sys)$. We recall that the topological dual $\I_1(\H_\sys)^*$ can be identified with $\B(\H_\sys)$, equipped with the operator norm~$\|A\|_\infty=\sup_{\psi\in \H_\sys, \|\psi\|\leq 1}\|A\psi\|$ through the duality
\[(\rho,A)\mapsto \tr(\rho A).\]
The trace-preservation of $\L_k$ is equivalent to $\L_k^*(\id)=\id$. The adjoint $\L_k^*$ is then a positive, unital linear map on the Banach space $\B (\H_\sys)$, and by the Russo-Dye theorem (\hspace{1sp}\cite{RD}), its operator norm $\|\L_k^*\|$ satisfies $\|\L_k^*\|=\|\L_k^*(\id)\|_{\infty}$ so that $\|\L_k\|=\|\L_k^*\|=1$. However, $\L_k$ is in general not a contraction when we equip the set $\B(\H_\sys)$ with the Hilbert-Schmidt norm $\|\cdot\|_2$ induced by the inner product $(A,B)\mapsto \tr(A^* B)$: it is immediate that a necessary condition is~$\L_k(\id)=\id$ (and this condition can be seen to be sufficient using the operator Schwarz inequality, see Theorem~5.3 in \cite{wolftour}). Note also the property $\L_k(\rho)^*=\L_k(\rho^*)$, 
from which we get that the spectrum of~$\L_k$ on $\I_1(\H_\sys)$ is symmetric with respect to the real axis: $\sp (\L_k)=\overline{\sp (\L_k)}$.

\paragraph{Example}
We consider the simplest non-trivial RIS, in which both the system and probes are 2-level systems, i.e. $\H_\sys = \H_\env =\C^2$. Moreover, we specify Hamiltonians $\hsys = E \,a^* a$ and $\henv = E_0\, b^* b$, where $a$/$a^*$, resp. $b$/$b^*$, are the annihilation/creation operators for $\sys$, resp. $\env$. As matrices expressed in the (ground state, excited state) bases of each system,
\begin{align*}
	a = b &= \begin{pmatrix} 0& 1 \\ 0 & 0 \end{pmatrix},
		& a^* = b^* &= \begin{pmatrix} 0& 0 \\ 1 & 0 \end{pmatrix}, &
	a^* a = b^* b &= \begin{pmatrix} 0& 0 \\ 0 & 1 \end{pmatrix}.
\end{align*}
Our first example will be when the two systems are coupled through their dipoles, in the \textit{rotating wave approximation}. This means that the system and chain elements interact via a constant potential $\lambda v_\text{RW}$, where 
\[ v_\text{RW} = \frac{u_1}{2}(a^* \otimes b + a \otimes b^*),\]
where $u_1$ is a constant, which we take equal to $1$ with units of energy.
This is a common approximation in the regime $|E-E_0| \ll \min\{E,E_0\}$ and $\lambda \ll |E_0|$.
To discuss the application of our results, we will also consider the full dipole interaction without the rotating wave approximation, that is, we assume a constant potential $\lambda v_\text{FD}$ with
\[v_\text{FD} = \frac{u_1}{2} (a^* + a) \otimes (b^* + b) = \frac{u_1}{2}(a \otimes b + a^* \otimes b + a \otimes b^* + a^* \otimes b^*).\]
We do not discuss the relative merits of these two models of interaction, but use them to describe two simple examples of repeated interaction systems that have different features with respect to Landauer's Principle in the adiabatic limit (see  Sections\nobreakspace \ref {sec:rotwave} and\nobreakspace  \ref {sec:fulldip}).

Before we move on to the next section, let us introduce a variant of a RIS described by the constant elements $\H_\sys$, $\hsys$, $\H_\env$ and the collection of variables $(\henv[k])_{k=1,\ldots,T}$, $(v_{k})_{k=1,\ldots,T}$,  $(\beta_k)_{k=1,\ldots,T}$. We define the $m$-\emph{repeated version} of this RIS ($m\in\nn$), which we call a $m$-RIS, to be the one associated with $\H_\sys$, $\hsys$, $\H_\env$ and $(\henv[[(k'-1)/m]+1])_{k'=1,\ldots,mT}$, $(v_{[(k'-1)/m]+1})_{k'=1,\ldots,mT}$,  $(\beta_{[(k'-1)/m]+1})_{k'=1,\ldots,T}$. The $m$-repeated version is simply obtained from the original RIS by considering $m$ identical copies of the $k$-th probe before moving on to another probe. This repeated version of an RIS will be useful in later applications.

\subsection{Landauer Principle in RIS} \label{ssec:LanRIS}
Each step of the repeated interaction system consists initially of a product state between the system $\sys$ and a thermal state $\envstate$, which are then time evolved unitarily. This is the framework of Landauer's Principle, as stated in~Section\nobreakspace \ref {sec:Landauer},  so the balance equation, Equation\nobreakspace \textup {(\ref {eq:balance})}, holds at each step. That is, if $\Delta S_k$ is the decrease in entropy of the system during step $k$, and $\Delta Q_k$ is the increase in energy of the probe $\env_k$, i.e.
\begin{align}
	\Delta S_k &:= S(\sysstate[k-1]) - S(\sysstate[k]) = S(\sysstate[k-1]) - S(\L_k(\sysstate[k-1])),\\
	\Delta Q_k &:= \tr\big( \henv[k] \tr_\sys\!\big(U_k (\sysstate[k-1] \otimes \envstate[k]\init) U_k^*\big)\big) - \tr(\henv[k] \envstate[k]\init),
\end{align}
then from Section\nobreakspace \ref {sec:Landauer} we have the entropy balance equation:
\begin{equation}\label{eq:RISbal}
\Delta S_k +\sigma_k = \beta_k \Delta Q_k
\quad \text{
with
} \quad
\sigma_k := S\big(U_k (\sysstate[k-1] \otimes \envstate[k]\init) U_k^* | \L_k(\sysstate[k-1])\otimes \envstate[k]\init \big).
\end{equation}
We will simplify notation and write $\envstate[k]$ for $\envstate[k]\init$. If we consider the first $T$ steps of the repeated interaction process, we sum over $k$ to obtain
\begin{equation}\label{eq_LandauerRIS}	
\sum_{k=1}^T \Delta S_k + \sum_{k=1}^T \sigma_k = \sum_{k=1}^T \beta_k \Delta Q_k.
\end{equation}
In particular we have the repeated interaction system Landauer inequality
\begin{equation}\label{eq_LandauerInequalityRIS}	
S(\sysstate[0]) - S(\sysstate[T]) = \sum_{k=1}^T \Delta S_k\leq  \sum_{k=1}^T \beta_k \Delta Q_k.
\end{equation}
This inequality will be saturated if and only if $\sigma_k=0$ for all $k$. Our main interest in the rest of this paper will be to discuss the behaviour of $\sum_{k=1}^T \sigma_k$ in the adiabatic limit, and in particular to characterize the saturation. Note that, by adiabatic limit, we mean here that parameters $h_{\env_k}, \beta_k, v_k$ will change slowly between probes; the model, however, remains a discrete time model with instantaneous changes.
\begin{remark}
As an alternative approach, one can consider the larger Hilbert space consisting of the  chain up to step $T$ and the small system. If we assume that a probe $\env_j$ not currently interacting with the small system evolves following its free Hamiltonian $\henv[j]$, the time evolution during step $k$ is given by 
\begin{align*}	
\tilde{U}_k &= e^{-i\tau_1 \henv[1]}\otimes \dotsm\otimes e^{-i\tau_{k-1}\henv[k-1]}\otimes U_k \otimes e^{-i\tau_{k+1}\henv[k+1]}\otimes \dotsm \otimes e^{-i\tau_T \henv[T]}
\end{align*}
where $U_k$ is the evolution between the system $\sys$ and $\env_k$, defined by \eqref{eq_defUk}.
Denote $\tilde U _{\mathrm{tot}} = \tilde U_T \ldots \tilde U_1$
and for $\rho$ a state on $\mathcal H_\sys$ let
\[\mathcal L_{\mathrm{tot}}(\rho)= \tr_{\env_1,\ldots,\env_T} \big(\tilde U_{\mathrm{tot}}(\rho\otimes \xi_{\env_1,\ldots,\env_T})\tilde U_{\mathrm{tot}}^*\big)\]
where the partial trace is now defined by $\tr_{\env_1,\ldots,\env_T}(A\otimes B_1\otimes \ldots\otimes B_T):=\tr(B_1\otimes \ldots\otimes B_T)\, A$,  and where we use the shorthand $\xi_{\env_1,\ldots,\env_T}=\bigotimes_{j=1,\ldots,T} \xi_j$. Then by a direct computation we have
\[ \sum_{k=1}^T \sigma_k = S\big(\tilde U_{\mathrm{tot}}(\rho\init\otimes\xi_{\env_1,\ldots,\env_T}) \, \tilde U_{\mathrm{tot}}^*\, | \,\mathcal L_{\mathrm{tot}}(\rho\init)\otimes \xi_{\env_1,\ldots,\env_T}\big).\]

Assume now that, instead of a repeated interaction system, we have a single step with unitary evolution given by $\tilde U_{\mathrm{tot}}$. Then we are in the situation described by Section\nobreakspace \ref {sec:Landauer} and the final state of the small system is
\[\sysstate^{\textnormal{1-step}}_T:=\mathcal L_{\mathrm{tot}}(\rho\init) = \sysstate[T].
\]
A slight variation on the proof of the entropy balance equation given in Section\nobreakspace \ref {sec:Landauer} gives
\[ \Delta S^{\textnormal{1-step}} + \sigma^{\textnormal{1-step}}= \sum_k \beta_k \Delta Q_k,\]
where we have, as one could expect,
\begin{gather*}
\Delta S^{\textnormal{1-step}}:=S(\sysstate\init) - S(\sysstate^{\textnormal{1-step}}_T) =\sum_{k=1}^T \Delta S_k,\\
\sigma^{\textnormal{1-step}}_{\mathrm{tot}} = S\big(\tilde U_{\mathrm{tot}}(\rho\init\otimes\xi_{\env_1,\ldots,\env_T}) \, \tilde U_{\mathrm{tot}}^*\, | \,\mathcal L_{\mathrm{tot}}(\rho\init)\otimes \xi_{\env_1,\ldots,\env_T}\big) = \sum_k \sigma_k.
\end{gather*}
One could therefore hope, in analogy with the situation described in Section\nobreakspace \ref {subsec_adiabatic}, to derive the vanishing of $\sum_k \sigma_j$ from an adiabatic-type result showing the convergence of $\tilde{U}_{\mathrm{tot}}(\rho\init\otimes\xi_{\env_1,\ldots,\env_T}) \, \tilde U_{\mathrm{tot}}^*$ to $\mathcal L_{\mathrm{tot}}(\rho\init)\otimes \xi_{\env_1,\ldots,\env_T}$. However, $\tilde{U}_{\mathrm{tot}}$ involves the first $T$ environments, so that controlling its spectral properties as $T\to\infty$ is difficult. In addition, the state of the joint system consisting of~$T$ environments is in general not a KMS state, so that we cannot use spectral characterizations derived from Araki theory.
Note that these difficulties in adapting a proof strategy from the continuous time, single reservoir case reflect the intrinsic differences of the systems. The discrete time dynamics of interacting with a chain of probes with possibly varying parameters involves fundamentally different physics than the return to equilibrium behaviour of interactions with a single thermodynamic reservoir. 
Our approach will therefore be to use an adiabatic result for the reduced dynamics, acting on $\mathcal H_\sys$ alone.
 \end{remark}

Last, to a $m$-repeated version of a RIS we associate for $k\in\nn$, $j=1,\ldots,m$ the quantities

\begin{equation} \label{eq_mrepeatedobs}
	\Delta^{(j)} S_{k,T} = \Delta S_{(k-1)m + j},\quad \Delta^{(j)} Q_{k,T} = \Delta Q_{(k-1)m + j}, \quad \sigma^{(j)}_{k,T} = \sigma_{(k-1)m + j},
\end{equation}	
which are the changes in entropy, energy, and the entropy production, during the interaction with the $j$-th copy of probe $k$.

\subsection{Strategy} \label{sec:strat}
We finish this section with some remarks regarding our strategy to study the saturation of Landauer's inequality \eqref{eq_LandauerInequalityRIS}. 
We need to estimate the entropy production of the $k$-th step, $\sigma_k$. Define 
\begin{equation} \label{eq_expsigmak}
\omega_{U,k}:=U_k (\sysstate[k-1] \otimes \envstate[k]) U_k^*, \quad \omega_{\L,k}= \L_k(\sysstate[k-1])\otimes \envstate[k],\quad \mbox{so that} \quad \sigma_k = S(\omega_{U,k}\,|\, \omega_{\L,k}).
\end{equation}
To estimate $\sigma_k$, we need estimates on the initial and final system states for step $k$, that is $\sysstate[k-1] = \L_{k-1} \dotsm \L_1 (\sysstate\init)$ and $\sysstate[k]=\L_k(\sysstate[k-1])$. According to the general picture of Landauer's Principle (see Section\nobreakspace \ref {sec:Landauer}), the relevant regime is the adiabatic limit. In that regime, $\sysstate[k]$ and $\sysstate[k-1]$ should be close to one another, and we investigate the difference $\sysstate[k]-\sysstate[k-1]$. A relevant adiabatic theorem allowing us to control this difference is developed in Section\nobreakspace \ref {sec:Adiabatic}. We then need perturbative estimates for the relative entropy $S(\rho|\nu)$  of two states with $\rho-\nu$ small. Such estimates are given in Section\nobreakspace \ref {sec:Relative-entropy}. Last, as we discussed in Section\nobreakspace \ref {sec:Landauer}, the results of \cite{JP14} showing saturation in the adiabatic limit use an ergodic assumption on the different dynamics of the system. In the current framework the assumption of ergodicity for a system with dynamics induced by the CPTP maps $\L_k$ will be irreducibility. This is  consistent with the standard definition of ergodicity for CPTP maps, which is essentially that the eigenvalue $1$ is simple and the associated eigenvector has full rank (see the review paper \cite{Sch}), but not with the ergodicity assumptions used in the RIS review \cite{BJM14}, which is only that the eigenvalue 1 is simple. We discuss this in more detail when we study the application of the results from Section\nobreakspace \ref {sec:Adiabatic} and Section\nobreakspace \ref {sec:Relative-entropy} to Repeated Interaction Systems in Section\nobreakspace \ref {sec:Application}. We also address there the small coupling limit $\lambda\to 0$ of our results, in order to get asymptotics of entropy production. This requires some technical considerations treated in Section\nobreakspace \ref {sec:scl}.

\section{The discrete non unitary adiabatic theorem} \label{sec:Adiabatic}

In this section, we look for an adiabatic approximation of a repeated interaction system. More precisely, we consider a RIS from time $0$ to time $T$, where the $\L_k$ change slowly; that is, when each $\L_k$ is $\L_{k,T}=\L(k/T)$ for $[0,1]\ni s\mapsto \L(s)$ a smooth enough function. According to the typical adiabatic framework described in Section\nobreakspace \ref {subsec_adiabatic}, we would like to find a family of operators $(A_{k,T})_{k=1,\ldots,T}$ such that each $A_{k,T}$ maps a given spectral subspace of $\L_{0,T}$ to the corresponding spectral subspace of $\L_{k,T}$, and $A_{k,T}$ is a good approximation for $\L_{k,T}\ldots \L_{1,T}$.

However, most adiabatic results (for example \cite{BornFock,KatoPaper,AE99,Teufel} in continuous time, or \cite{DKS98,Tan11} in discrete time) apply to unitary dynamics only, whereas here the $\L_{k,T}$ are not unitary. On the other hand, adiabatic results for non unitary dynamics (see \cite{ASF,Joye,AFGG,Schmid}) are proven only in continuous time, whereas we work here in discrete time. We will therefore need to derive a specific form of adiabatic theorem for products of slowly changing CPTP maps. This will prevent us from finding unitary $A_{k,T}$, and we will only have approximately unitary $A_{k,T}$, in a sense to be made precise below.
 
In this section, we first formulate the minimal hypothesis for the theorem in an abstract setup, before discussing possible relaxation of these assumptions and immediate consequences of these hypotheses.  We state our main theorem in Theorem\nobreakspace \ref {theo_adiabatic}. The strategy of our proof and related considerations are discussed after the statement of the theorem. Note that in the following abstract setup, we simply require the maps to be contractions instead of CPTP, though we do investigate the consequences of those properties in~Section\nobreakspace \ref {sec:Application}.

\subsection{Setup, hypotheses and statement}
\label{sec:setup}

Let $X$ be a finite-dimensional Banach space equipped with a norm $\|\cdot\|$. See the remark at the end of the section for infinite dimensional $X$, though. We will denote similarly the induced operator norm on $\B(X)$. We recall that, for any bounded linear map $A$ on $\B(X)$, and any subset~$S$ of $\sp\,A$, one defines the \textit{spectral projector of $A$ on $S$} by the operator-valued complex integral
\begin{equation} \label{eq_defspectralprojector}
P_{S,A} = \frac1{2\i \pi}\int_{\Gamma} (z-A)^{-1} \, \d z
\end{equation}
with $\Gamma$ any simple closed contour with positive orientation, encircling $S$ and none of $\sp\,A\setminus S$.

We are now ready to state our main assumptions. Consider an operator valued function $[0,1] \ni s \mapsto \L(s) \in \B(X)$, and let $S^1$ be the unit circle in $\C$. 
\begin{remark}
	Here and in what follows, we say a function $f$ on $[0,1]$ is $C^2$ on $[0,1]$ if $f$  is continuous  on $[0,1]$, twice continuously differentiable on $(0,1)$, and $\lim_{s \downarrow 0} f'(s)$, $\lim_{s \uparrow 1} f'(s)$, $\lim_{s \downarrow 0} f''(s)$ and $\lim_{s \uparrow 1} f''(s)$ exist and are finite.
\end{remark}
 We consider the following assumptions:
\begin{itemize}
	\item[\textbf{\textup{H1}.}] For each $s \in [0,1]$, $\L(s)$ is a contraction, i.e. $\|\L(s)\|\leq 1$. \label{Hcontraction}
	
	\item[\textbf{\textup{H2}.}] There is a uniform gap $\epsilon>0$ such that, for $s \in [0,1]$, each \textit{peripheral eigenvalue} $e^j(s)\in \sp\L(s)\cap S^1$ is simple, and $|e^j(s)-e^i(s)| > 2\epsilon$ for any $e^j(s)\neq e^i(s)$ in $\sp\L(s)\cap S^1$.
	
	\item[\textbf{\textup{H3}.}] Let $P^j(s)$ be the spectral projector associated with $e^j(s) \in \sp\L(s) \cap S^1$, and $P(s)=\sum_j P^j(s)$ the \textit{peripheral spectral projector}. The map $s\mapsto \L^P(s) := \L(s) P(s)$ is $C^2$ on $[0,1]$. 
	
	\item[\textbf{\textup{H4}.}]  There is a uniform bound on the strictly contracting part of $\L(s)$, i.e. if $Q(s) := \one - P(s)$, 
	\begin{align*}	
	\ell:=\sup_{s\in[0,1]}\|\L(s)Q(s)\| < 1.
	\end{align*}
\end{itemize}
\begin{remark} As is well known, the eigennilpotents associated to finite dimensional peripheral eigenvalues of contractions are equal to zero. Hence \textup{H1} implies that for each $s$, $\L^P(s)$ is simple in the sense that $\L^P(s) = \sum_j e^j(s) P^j(s)$. We denote by $N(s)$ the number of distinct peripheral eigenvalues of $\L(s)$. Observe that, by assumption \textup{H2} one has
\begin{equation} \label{eq_unifboundNs}
N_{\max}:=\sup_{s\in[0,1]}N(s) \leq \min(\frac{2\pi}\epsilon, \dim\,X).
\end{equation}
\end{remark}

In the applications, $X$ will be the Banach space of trace-class operators on a Hilbert space $\H$, equipped with the trace norm $\|A\|_1 := \tr \sqrt{A^* A}$. We recall that we denote by $\|\cdot\|$ the operator norm on $\mathcal B(X)$, as e.g. in \textup{H1} and \textup{H4}. We will call $s$ the \textit{dimensionless time parameter}. For~$T \in \N$ and $0 \leq k \leq T$, we set \begin{equation} \label{eq_defLkT}
\L_{k,T} := \L(k/T), \qquad P_{k,T} := P(k/T), \qquad e_{k,T}^j := e^j(k/T). 
\end{equation}
We denote simply $\L_0=\L(0)$, $P_0=P(0)$, $e_{0}^j=e^j(0)$ to emphasize the fact that these quantities do not depend on $T$. The parameter $T$ plays the role of a time scale, which we will call the \textit{adiabatic parameter}. Our goal will be to prove that,
 for $s\in [0,1]$ the evolution $\L_{[sT],T}\ldots \L_{0,T}$ maps every $\ran P^j(0)$ to $\ran P^j(s)$ in the adiabatic limit $T \to \infty$.
\medskip

Assume now that the map $s\mapsto \L(s)$ satisfies \textup{H1}, \textup{H2}, \textup{H3}, \textup{H4}, and define $\L_{k,T}$ by \eqref{eq_defLkT}. For the rest of this section, we neglect the subscript $T$ and simply denote e.g. $\L_k= \L_{k,T}$. We now discuss immediate consequences of the above framework.

The spectral projectors of $\L_k$ commute with $\L_k$, so that we have a decomposition
\begin{equation*}
	\L_k = \L^P_k + \L^Q_k \mbox{ where } \L^P_k := \L_kP_k \mbox{ and } \L^Q_k:=\L_k Q_k.
\end{equation*}
Note that if $\L_k$ has $\sp(\L_k)\cap S^1 = \{1\}$, then $ \L_k P_k = P_k$, the eigenprojection corresponding to 1. We have two brief lemmas about the eigenprojectors and eigenvalues of $\L^P$.

\begin{lemma} \label{lemma_UnifBoundLPn}
Assume that $s\mapsto \L(s)$ satisfies \textup{H1}, \textup{H2} and \textup{H4}, and let $P_k^j$ be the eigenprojector corresponding to a peripheral eigenvalue $e^j_{k}$ of $\L_k$ and $P_k = \sum_j P^j_k$. Then
\begin{enumerate}
\item for each $j$, $\|P_k^j\|=1$, and $\|P_k \|=1$, so that $\|P_k\L_k\|\leq 1$ and $\|Q_k\|\leq 2$,
\item if in addition $\L_k$ is CPTP, then both $P_k$ and $P_k \L_k$ are CPTP.
\end{enumerate}
\end{lemma}
The proof is given in Appendix\nobreakspace \ref {App:Tanaka}.

\begin{remarks}\hfill \vspace{-0.6em}
\begin{itemize}
\item Since the norm 
on $X$ is not necessarily induced by an inner product, the projectors $P^j$ are not necessarily self-adjoint. See Section\nobreakspace \ref {sec:Application} for an example. 
\item The arguments for part 2 of Lemma\nobreakspace \ref {lemma_UnifBoundLPn} are borrowed from \cite{wolftour}.
\end{itemize}
\end{remarks}

\begin{lemma} \label{lem:eigC2}\label{lem:projC2}
Assume that $s\mapsto \L(s)$ satisfies \textup{H1}--\textup{H4}. The eigenvalues $e^j(s)$ and eigenprojectors $P^j(s)$ of $\L^P(s) = \L(s)P(s)$ are $C^2$ as functions of $s$ on $[0,1]$. In particular, there exists a constant $c^P$ such that for all $s\in[0,1]$,
\begin{equation} \label{eq_defCPCPp}
\max_{\alpha=1,2} \max_{j} \big(\Big\| \frac{\d^\alpha P^j(s)}{\d s^\alpha} \Big\|, \Big|\frac{\d^\alpha e^j(s)}{\d s^\alpha}\Big|\big) \leq c^P. 
\end{equation}
\end{lemma}
See Appendix\nobreakspace \ref {App:Tanaka} for the proof. Note that as a consequence, $N_{\max} \equiv N(s)$ in \eqref{eq_unifboundNs}.

\begin{remark}
The above conditions \textup{H1}--\textup{H4} can usually be relaxed if we allow for a modification of the framework. In particular, the norm condition \textup{H4} can be replaced by a spectral radius condition, at the cost of replacing $\L$ by $\L^m$ for some power $m$. This is necessary for applications, and in the small coupling limit discussed in Section\nobreakspace \ref {sec:scl}. More precisely, define the following weaker version of \textup{H4}:
\smallskip

\noindent \textup{\textbf{\,\textup{wH4}.}} We have the uniform spectral bound $\ell':=\sup_{s\in [0,1]}\, \spr \L(s) Q(s) < 1$.
\smallskip

We can then prove the following easy result:
\begin{lemma} \label{lemma_wH4H4}
Assume that $s\mapsto \L(s)$ is continuous on $[0,1]$ and satisfies \textup{H1}, \textup{H2}, \textup{H3} and \textup{wH4}. For any $\ell\in {]\ell',1[}$, there exists $m_0$ in $\N$ such that, for any $m\geq m_0$, the map $[0,1]\ni s \mapsto \L(s)^m$ satisfies \textup{H1}, \textup{H3} and \textup{H4}.
\end{lemma}

\begin{proof}
Let $\epsilon_0>0$ such that $\ell=\ell'+ 2\epsilon_0$ and let $s\in[0,1]$. Since \[\spr(\L(s)Q(s)) = \lim_{m\to\infty} \|  \L(s)^m Q(s)\|^{1/m} \leq \ell',\] there is some $m(s)$ such that $\| \L(s)^m Q(s)\| \leq (\ell'+\epsilon_0)^m \leq \ell' + \epsilon_0$ for $m\geq m(s)$. Since each projector onto eigenvalues of $\L(s)$ on the unit circle is $C^2$, so are projector $P(s)$ and $\one-P(s)=Q(s)$, and therefore $s\mapsto  \L(s)^m Q(s)$ is continuous. Then there exists an open interval~$I_s\ni s$ such that if $s'\in I_s$, $\| \L(s')^{m(s)} Q(s')\| \leq \ell' + 2 \epsilon_0$. For $m\geq m(s)$ we still have for~$s'\in I_s$, \[\| \L(s')^{m} Q(s')\| \leq \|\L(s')^{m-m(s)}\| \, \| \L(s')^{m(s)} Q(s')\|\leq \ell' + 2 \epsilon_0=\ell\]
and considering a finite cover $I_{s_1}$, \ldots $I_{s_p}$ of $[0,1]$ we can take $m_0=\max \{m(s_1),\ldots, m(s_p)\}$. In addition, $\| \L^m \| \leq \| \L\|^m$ and $(\L^m)^P= (\L^P)^m$, so we have \textup{H1} and \textup{H3}.
\end{proof}
However, the operator $\L(s)^m$ may have degenerate eigenvalues on the unit circle for arbitrarily large $m$, and the way such $m$ depends on $s$ is non-trivial. We will see, however, that when the map $s\mapsto \L(s)$ derives from a RIS and satisfies an ergodicity property, we have additional information about the peripheral spectrum of $\L(s)$, and if $\L$ satisfies \textup{H1}, \textup{H2}, \textup{H3} and \textup{wH4} then we can find $m$ such that $\L^m$ satisfies \textup{H1}--\textup{H4}. We discuss this in Section\nobreakspace \ref {sec:Application}.
\end{remark}
\medskip

We are now ready to state our adiabatic theorem.
\begin{theorem} \label{theo_adiabatic}
Under \textup{H1}--\textup{H4}, there exist constants $T_0>0$ and $C^P>0$, depending only on $c^P$ defined by \eqref{eq_defCPCPp}, and on $N_{\mathrm{max}}$ defined by \eqref{eq_unifboundNs} such that
for  all $T\geq T_0$, there exist two families of maps $(A_{k,T})_{k=1,\ldots,T}$ and $(A^\dagger_{k,T})_{k=1,\ldots,T}$, with uniform bounds 
\begin{equation} \label{eq_unifboundA}
\sup_{k=0,\ldots,T}\max(\|A_{k,T}\|, \|A_{k,T}^\dagger\|)\leq 
N_{\mathrm{max}} (1-\frac{(c^P)^2}{T_0^2})^{-T_0/2}
\end{equation}
satisfying
\begin{gather*}
A_{k,T}^\dagger \, A_{k,T}= P_0,\qquad  A_{k,T}\, A_{k,T}^\dagger = P_{k,T}, \\
A_{k,T}\, P_0^j = P_{k,T}^j \, A_{k,T}, \qquad A_{k,T}^\dagger\, P_{k,T}^j = P_0^j \, A_{k,T}^\dagger, \qquad A_{k,T} Q_0=Q_0 A^\dagger_{k,T}=0,
\end{gather*}
such that, for all $k\leq T$,
\begin{equation} \label{eq_adiabatic1}
	\big\|\L_{k,T}\L_{k-1,T}\ldots \L_{1,T} - A_{k,T}\big\| \leq \frac {C^P}{T(1-\ell)} + 2\ell^k.
\end{equation}
\end{theorem}

Theorem\nobreakspace \ref {theo_adiabatic} follows from two distinct results that are Propositions\nobreakspace \ref {prop:Qkilled} and\nobreakspace  \ref {prop:general-ad}. Proposition\nobreakspace \ref {prop:Qkilled} shows that in the adiabatic regime, the $\L^Q$ terms in the expansion of $\L_{k,T}\ldots \L_{1,T} P_0$ can be neglected, more precisely:
\begin{equation} \label{eq_adiabatic2}
	\|\L_{k,T}\ldots \L_{1,T}P_0 - \L_{k,T}^P\ldots \L_{1,T}^P P_0\|\leq  \frac{C^P}{T(1-\ell)}, \text{ and } \|\L_{k,T} \dotsm \L_{1,T} Q_0\|\leq  \frac{C^P}{T(1-\ell)} + 2\ell^k.
\end{equation}
Then Proposition\nobreakspace \ref {prop:general-ad} proves that $\L_{k,T}^P\ldots \L_{1,T}^P P_0$ is well approximated by the operator $A_{k,T}$, which is described in Section\nobreakspace \ref {sec:unitary}. Note that the bounds \eqref{eq_adiabatic2} give more information than \eqref{eq_adiabatic1}; in the sequel we will need \eqref{eq_adiabatic2}.

\subsection{Bounding the strictly contracting part of $\L$} \label{sec:strict}

We first have a proposition which shows that the strictly contracting part of $\L$ does not contribute to adiabatic evolution when the initial state is in $P_0 X$. 

\begin{proposition} \label{prop:Qkilled}
If $\L(s)$ satisfies \textup{H1}--\textup{H4}, then there exists a constant $C^P$, depending only on~$c^P$ defined by \eqref{eq_defCPCPp}, such that for any $T\geq 1$ and $k\leq T$, 
\begin{equation}	
\| \L_{k,T} \L_{k-1,T}\dotsm\L_{1,T} - \L^P_{k,T} \L^P_{k-1,T}\dotsm  \L^P_{1,T}P_0 - \L^Q_{k,T} \L^Q_{k-1,T}\dotsm  \L^Q_{1,T} Q_0\| \leq \frac{C^P}{T(1-\ell)}, \label{eq:Qkilled}
\end{equation}
where $\L^P_{k,T} = \L_{k,T} P_{k,T}$ and $\L^Q_{k,T} =\L_{k,T} Q_{k,T}$. Moreover, $\|\L^Q_{k,T} \L^Q_{k-1,T}\dotsm  \L^Q_{1,T} Q_0\| \leq 2\ell^k$.
\end{proposition}
We postpone the proof to Appendix\nobreakspace \ref {App:Tanaka}.

\begin{remarks}\nopagebreak \hfill \vspace{-0.6em}
\begin{itemize}
\item From Equation\nobreakspace \textup {(\ref {eq:Qkilled})}, one has $\max(\|P_k\L_k \dotsm \L_1Q_0 \|,\|Q_k \L_k \dotsm \L_1 P_0\|)\leq \frac{C^P}{T(1-\ell)}$.
\item As the proof shows, one cannot obtain a better dependency on $T$ than $1/T$ using this method: the expansion of \eqref{eq_LminusLP} contains terms of the form $\L^Q_k \dotsm \L^Q_{n+1} \L^P_n \dotsm \L^P_1 P_0$ which have exactly one $Q_nP_{n-1}$ part to pick up a $1/T$ factor. Remark in addition that there are $k-1$ such terms; since we can have $k=T$, our claim would break down without the assumption $\|\L^Q\|\leq \ell < 1$.
\item Similarly, the dependence on $\ell$ of 
the result 
cannot be improved substantially: indeed, expression \eqref{eq_l_optimal1} is bounded below by the term $d=1$, $\frac{ c }T\sum_{\alpha=1}^{k-1} \ell^\alpha = \frac{  c \ell(1-\ell^{k-1})}{T(1-\ell)}$ which already captures the essentials of the upper bound.
\end{itemize}
\end{remarks}

\subsection{Approximating $\L^P_k \L^P_{k-1} \dotsm \L^P_1 P_0$ by an adiabatic dynamics}
	
\label{sec:unitary}

Define two families $(\kat{k,T})_{k=0,\ldots,T}$ and $(\kat{k,T}^\dagger)_{k=0,\ldots,T}$ by $\kat{0,T} = \kat{0,T}^\dagger = P_{0,T}$, and
\begin{equation}\label{eq:kat}
\begin{aligned}
\kat{k+1,T} := \sum_j P_{k+1,T}^j P_{k,T}^j \big(\one-(P_{k+1,T}^j-P_{k,T}^j)^2\big)^{-1/2},  \\
\kat{k+1,T}^\dagger := \sum_j  P_{k,T}^j P_{k+1,T}^j \big(\one-(P_{k+1,T}^j-P_{k,T}^j)^2\big)^{-1/2}.
\end{aligned}
\end{equation}
These families are well-defined for large enough $T$ by Lemma\nobreakspace \ref {lem:projC2}. The observation that every $(P_{k+1,T}^j-P_{k,T}^j)^2$ commutes with both $P_{k,T}^j$ and $P_{k+1,T}^j$, and simple formulas like e.g. $P_{k,T}^j P_{k+1,T}^j P_{k,T}^j=P_{k,T}^j\big(\id - (P_{k+1,T}^j-P_{k,T}^j)^2\big)$ lead immediately to the intertwining relations:
\begin{gather*}
\kat{k+1,T} \, P_{k,T}^j = P_{k+1,T}^j\, \kat{k+1,T}, \qquad \kat{k+1,T}^\dagger\, P_{k+1,T}^j = P_{k,T}^j \,\kat{k+1,T}^\dagger,\\
\kat{k+1,T}^\dagger \,\kat{k+1,T} = P_{k,T}, \qquad \kat{k+1,T}\,\kat{k+1,T}^\dagger = P_{k+1,T},
\end{gather*}
(see \S 4.6 in chapter I of \cite{Kato} for similar results).

\begin{remark}
The operator $\kat{k+1,T}^\dagger$ is a pseudo-adjoint of $\kat{k+1,T}$, in the sense that we would have $\kat{k+1,T}^*=\kat{k+1,T}^\dagger$ if the spectral projectors $P_{k,T}^j$ were self-adjoint. We continue with this notation throughout this section, and every operator $Y^\dagger$ will be a pseudo-adjoint of $Y$, depending on $\{P_{k,T}^j\}_{k,j}$.
\end{remark}
 Now define $\Kat{0,T} = \Kat{0,T}^\dagger = \one$, and
\begin{equation}
\Kat{k,T}  := \kat{k,T}\ldots \kat{1,T},  \qquad \Kat{k,T}^\dagger := \kat{1,T}^\dagger\ldots \kat{k,T}^\dagger. \label{eq:Kat}
\end{equation}
By definitions Equation\nobreakspace \textup {(\ref {eq:kat})} and Lemma\nobreakspace \ref {lem:projC2} one has uniform bounds for $\Kat{k,T}$, $\Kat{k,T}^\dagger$ for all $T\geq T_0(c^P)$, 
\begin{equation}  \label{eq_unifboundK}
\sup_{k=0,\ldots,T}\max(\|\Kat{k,T}\|, \|\Kat{k,T}^\dagger\|)\leq N_{\mathrm{max}} (1-\frac{(c^P)^2}{T^2})^{-T/2}.
\end{equation}
The operators $ \Kat{k,T}$ and $\Kat{k,T}^\dagger$ account for the motion of the spectral projectors between step 1 and step $k$. They satisfy 
\begin{equation}\label{eq_Kat2}
\begin{aligned}
\Kat{k,T} \,P_0^j = P_{k,T}^j \,\Kat{k,T}, &\qquad \Kat{k,T}^\dagger\, P_{k,T}^j = P_0^j \,\Kat{k,T}^\dagger, \\
\Kat{k,T}^\dagger\,\Kat{k,T}= P_0, &\qquad \Kat{k,T}\,\Kat{k,T}^\dagger = P_{k,T}.
\end{aligned}
\end{equation}
Now define  two families $(\Phase{k,T})_{k=1,\ldots,T}$ and $(\Phase{k,T}^\dagger)_{k=1,\ldots,T}$ by  $\Phase{0,T} = \Phase{0,T}^\dagger = P_{0,T}$, and
\begin{equation}
\Phase{k,T}  = \sum_j \big(\prod_{n=1}^{k} e^j_{n,T}\big)\, P_0^j,  \qquad
\Phase{k,T}^\dagger  = \sum_j \big(\prod_{n=1}^{k} \overline e^j_{n,T}\big)\, P_0^j. \label{eq:Phase}
\end{equation}
which satisfy $\Phase{k,T} \Phase{k,T}^\dagger = P_{0} = \Phase{k,T}^\dagger \Phase{k,T}$. We let $A_{k,T}=\Kat{k,T} \,\Phase{k,T}$ and $A^\dagger_{k,T}=\Phase{k,T}^\dagger \,\Kat{k,T}^\dagger$. These operators will be our \textit{adiabatic approximation} of the time evolution. Note that estimate (\ref{eq_unifboundA}) follows from (\ref{eq_unifboundK}) and the monotonicity in $T$ of its right hand side.

\begin{proposition} \label{prop:general-ad}
Under hypotheses \textup{H1}--\textup{H4}, there exists positive constants $T_0$ and $C^P$, depending only on $c^P$ defined by \eqref{eq_defCPCPp}, and on $N_{\mathrm{max}}$ defined by \eqref{eq_unifboundNs}, such that the adiabatic approximation~$(A_{k,T})_k$ defined above satisfies for all $T\geq T_0$ 
\begin{equation}
	\| \L^P_{k,T} \L^P_{k-1,T}\dotsm \L^P_{1,T} P_0 - A_{k,T}\|\leq C^P/T, \label{eq:katoapproxOfLp}
\end{equation}
where $A_{k,T} P_0^j = P_{k,T}^j A_{k,T}$, $A_{k,T} Q_0 = 0$, and $\|A_{k,T}\|$ is uniformly bounded in $k, T$ for $k\leq T$.
\end{proposition}

\begin{proof}
Our proof is a simple adaptation of  \cite{Tan11}. It is given in Appendix~\ref{App:Tanaka}.
\end{proof}

\begin{remark}
All statements of this Section hold for $X$ an {\it infinite} dimensional Banach space, assuming conditions \textup{H1}--\textup{H4}, if the differentiability conditions are understood in the norm sense. These assumptions imply that $\sup_{s}\mathrm{dim}\,P(s)X<+\infty$. 
\end{remark}

\section{Perturbation of relative entropy} \label{sec:Relative-entropy} \label{sec:perturb}
Following the strategy outlined in Section\nobreakspace \ref {sec:strat}, we now study an expansion of the relative entropy of two states that are both perturbations of the same state. This will allow us to estimate the entropy production \eqref{eq_expsigmak} of an RIS process in Section\nobreakspace \ref {sec:Application}. 

More precisely, let $\eta$ be a faithful state on a finite-dimensional Hilbert space $\H$, and let~$D_1$ and $D_2$ be two operators with $\tr D_j=0$, $j=1,2$. Our goal is to give an expansion of $S(\eta+D_1|\eta+D_2)$ to (combined) second order in $D_\ell$ in the sense that we will neglect terms of order $D_1^{\alpha_1}D_2^{\alpha_2}$ for $\alpha_1+\alpha_2 \geq 3$. As we will see in~Proposition\nobreakspace \ref {prop:perturbentropy}, this relative entropy vanishes to first order in the $D_j$, so we are simply expanding to the lowest non-vanishing order. This will follow from the (Dunford-Taylor) holomorphic functional calculus.
Consistently with the preceding sections, we denote by $\|\cdot\|$ the operator norm (of any linear map) and $\|\cdot\|_1$ the trace norm on $\I_1(\H)$.

\begin{proposition} \label{prop:perturbentropy}
Let $\eta$ be a faithful state with spectral decomposition $\eta = \sum_i  \mu_i p_i $, where $\mu_j$ are the eigenvalues and $p_j$ the associated spectral projections.
Let $D_1,D_2$ be two traceless perturbations of $\eta$. There exists a positive constant $C_\eta$ depending only on $\infspec\eta$,
\[C_\eta=c \, \dim \H \,\frac{\log2-\log\infspec\eta}{(\infspec\eta)^4}\]
(where $c$ is a numerical constant), and a constant $D_\eta > 0$ depending only on $\eta$ such that if $D_1,D_2$ satisfy $\|D_j\|\leq D_\eta$, $j=1,2$, then the relative entropy $S(\eta+D_1|\eta+D_2)$
satisfies
\begin{equation} \label{eq_perturbentropy}
\big| S(\eta+D_1|\eta+D_2) - F_\eta(D_1-D_2)\big| \leq C_\eta (\|D_1\|+\|D_2\|)^3
\end{equation}
where $F_\eta(A):= F_\eta(A,A)$ for
\begin{equation} \label{eq_defFeta}
	 F_\eta(A,B) := \sum_i\tr  (A p_i B p_i )\frac1{2\mu_i}  + \sum_{i<j} \tr(Ap_j Bp_i) \frac{\log(\mu_i) -\log(\mu_j)}{\mu_i-\mu_j}.
\end{equation}
\end{proposition}
The full proof of this proposition, including two technical lemmata, is given in Appendix\nobreakspace \ref {App:RelativeEntropy}.

\begin{remarks} We make a few comments on the implications of this proposition. From now on, for $Z$ a nonnegative quantity, we denote by $O_\eta(Z)$ any term that is bounded by $C_\eta \, Z$ for $Z$ small enough.

\begin{itemize}
\item
$F_\eta$ is a  bilinear form and it is not difficult to see using the Mean Value Theorem  that 
\begin{equation} \label{eq_boundFeta}
| F_\eta(A,B) |\leq \|A\| \|B\| \,\frac{(\dim {\H})^2}{\infspec\eta}.
\end{equation}
\item For later purposes, we note that if $\eta=\eta_0+\Delta$, so that $\tr \Delta=0$, we have
\begin{align}\label{perteta}
	S(\eta+D_1|\eta+D_2) = F_{\eta_0}(D_1-D_2) + O_{\eta_0}(\|D+\Delta\|^3),
\end{align}
with obvious notations.
\item For self-adjoint  $D_{1}$ and $ D_2$, as will be the case when both arguments of $S$ are states, we may take $D_\eta = \frac14\infspec\eta$, and  write
\begin{align*}
	&S(\eta+D_1|\eta+D_2)
	= \\
	&\ \ \ \ \sum_i \tr\big(|p_i(D_1\!-\!D_2)p_i|^2\big) \frac{1}{2\mu_i} + \sum_{ i < j} \tr\big(|p_i (D_1\!-\!D_2) p_j|^2\big) \frac{\log \mu_i - \log \mu_j}{\mu_i-\mu_j} + O_\eta(\|D\|^3).
	\end{align*}
	where both terms are non-negative, and the first is strictly positive given $D_1\neq D_2$. In the following, we will write $S(\eta+D_1|\eta+D_2) = O_\eta(\|D\|^2)$.

\item This proof also shows that if $D_1-D_2$ is not traceless, one obtains the same formula for $S(\eta+D_1|\eta+D_2)$ with the additional  term $\tr(D_1-D_2)$. 
	\item Note that on the other hand, by Theorem I.1.15 of \cite{OP}, if $\eta+D_1$, $\eta+D_2$ are states,
\begin{equation}\label{eq:quantum_pinsker}
	S(\eta+D_1|\eta+D_2) \geq \frac{1}{2}\|D_1-D_2\|_1^2.
\end{equation}
\end{itemize}
\end{remarks}

\section{Application to RIS} \label{sec:Application}

We now discuss some applications of the above results to repeated interaction systems. One of our tasks will be to verify the assumptions \textup{H1}--\textup{H4} for such systems. We will take advantage of the fact that operators $\L$ induced by repeated interaction systems are CPTP.

We recall some results about \emph{irreducibility} of CPTP maps and the related Perron-Frobenius theorem. For this assume, as will be the case in applications to RIS, that $X=\I_1(\H)$, with $\H$ a finite dimensional Hilbert space, equipped with the trace norm $\|\eta\|_1=\tr\big((\eta^*\eta)^{1/2}\big)$. We recall that a CPTP map $\L$ on $\I_1(\H)$ is a contraction and is called irreducible if the only self-adjoint projections $P \in \B(\H)$ satisfying $\L\big(P\I_1(\H) P\big) \subset P \I_1(\H) P$ are $P=0$ and $\id$. On top of the property $\sp (\L)=\overline{\sp (\L)}$, the Perron-Frobenius theorem for CPTP maps and related results (see \cite{EHK,Gro} or the reviews \cite{Sch,wolftour}) state that, if $\L$ is irreducible, then there exists $z$ in $\N$ such that
\begin{align}	\label{groupev}
 \sp(\L) \cap S^1 = S_z := \{ \exp(2\i \pi k/z), k=0,1,\ldots, z-1\},
\end{align} 
and each $\exp(2\i \pi k/z)$ is a simple eigenvalue. In addition, the (unique up to scalar multiplication) eigenvector for the eigenvalue 1 is positive definite, so that it can be chosen to be a state. In particular, there exists a unique invariant state $\sysstate\invar$.

If $\L$ is given in the form $\L(\eta)=\sum_{i\in I} V_i \eta V_i^*$ with $V_i\in \B(\H)$ for $i\in I$, then an equivalent condition for irreducibility is that there exist no non-trivial subspace of $\H$ left invariant by all operators $V_i$, $i\in I$; in particular, if $\dim \H=2$ then $\L$ is irreducible if and only if the operators $V_i$ do not have a common eigenvector.

We can now prove the following proposition. 
\begin{proposition} \label{prop:normTospr}
Assume $s\mapsto \L(s) \in \B(X)$ is a $C^2$ operator-valued function such that each $\L(s)$ is an irreducible CPTP map. If $s\mapsto \L(s)$ satisfies \textup{H4} then it satisfies \textup{H1}--\textup{H4}. If $s\mapsto \L(s)$  satisfies \textup{\textup{wH4}} then for any $\ell\in {]\ell',1[}$ there exists $m\in \nn$ such that $s\mapsto \L^m(s)$ satisfies \textup{H1}--\textup{H4}. In either case, there exists $z\in \mathbb N$ such that $\sp\L(s)\cap S^1 = S_z$ for all $s\in[0,1]$.
\end{proposition}
The proof relies simply on the rigidity induced by the group structure of the peripheral spectrum and the continuity in $s$ of the peripheral eigenvalues. It is given in Appendix\nobreakspace \ref {App:Application}. This implies in particular that $N_{\mathrm{max}}\equiv z$. We therefore drop all reference to $N_{\mathrm{max}}$ in this section.

In the context of RIS setup, for each step $k = 0, \dotsc, T$, one must choose the free Hamiltonian for the probe, $\henv[k]$, the inverse temperature of the probe $\beta_k$, as well as the interaction $v_k$  (in this section we assume constant interaction time $\tau > 0$ and coupling strength $\lambda$). We define the following asumption about these choices:
\smallskip

 \noindent \textbf{ADRIS} The RIS is associated with
\[\ham_{\env_k}=\ham_{\env_k,T}:= \ham_\env\Big(\frac kT\Big), \qquad \beta_{k}=\beta_{k,T}:=\beta\Big(\frac kT\Big), \qquad  v_{k}=v_{k,T}:=v\Big(\frac kT\Big),\]
for $k = 1, \dotsc, T$, where $s\mapsto \henv(s)$, $\beta(s)$, $v(s)$ are $C^2$ functions on $[0,1]$.
\smallskip
More explicitly, \textbf{ADRIS} means that we have 
\begin{equation}
	\begin{aligned} \label{eq_versionscontinues}
		U(s)&=\exp -\i \tau \big(\ham_\env(s)\otimes \id_{\env} + \id_{\sys}\otimes h_\env(s) + \lambda v(s)\big)\\
		\L(s)(\rho)&=\tr_\env\Big( U(s)\big(\rho\otimes \envstate[s]\big)U(s)^*\Big),
	\end{aligned} 
\end{equation}
where $\envstate[s]$ is the Gibbs state associated with temperature $\beta(s)$ and Hamiltonian $h_\env(s)$. Then $s \mapsto \L(s)$ is a  $\B(\H_\sys)$-valued $C^2$ function, and $\L_{k}=\L_{k,T}:=\L(k/T)$. We will then apply Theorem\nobreakspace \ref {theo_adiabatic} to such maps that are ergodic and satisfy \textup{H4} to characterize the entropy production in Proposition\nobreakspace \ref {prop:1atomentropy}. We define the following variant of ADRIS, which will be relevant in the case the family $\L(s)$ satisfies \textup{wH4} only:
\smallskip

\noindent \textbf{\textit{m}ADRIS} The repeated interaction system is the $m$-repeated version (see Section\nobreakspace \ref {ssec:RISmathdesc} and Figure\nobreakspace \ref {fig:sampling}) of a system satisfying ADRIS, i.e. we have for $k'=1, \dotsc, mT$,
\[\ham_{\env_k'}=\ham_{\env_k',T}:= \ham_\env\Big(\frac {[\frac{(k'-1)}{m}]+1}T\Big), \quad \beta_{k'}=\beta_{k',T}:=\beta\Big(\frac {[\frac{(k'-1)}{m}]+1}T\Big), \quad  v_{k'}=v_{k',T}:=v\Big(\frac {[\frac{(k'-1)}{m}]+1}T\Big),\]
where $s\mapsto \henv(s)$, $\beta(s)$, $v(s)$ are $C^2$ functions on $[0,1]$.

Under $m$-ADRIS, $\L_{k'}=\L_{k',T}$ is of the form $\L(([\frac{(k'-1)}{m}]+1)/T)$ with $s \mapsto \L(s)$ as in \eqref{eq_versionscontinues}.
\medskip

We can now state a result regarding the entropy production of a single step of an ADRIS system,
\begin{equation}
\sigma_{k,T} := S\big(U_{k,T} (\sysstate[k-1,T] \otimes \envstate[k,T]) U_{k,T}^* \,|\, \L_{k,T}(\sysstate[k-1,T])\otimes \envstate[k,T] \big).
\end{equation}

\begin{proposition}[1-RIS]\label{prop:1atomentropy}
Consider a repeated interaction system satisfying assumption \textup{ADRIS}, where the induced CPTP map $\L(s)$ is irreducible for all $s\in[0,1]$, and satisfies \textup{H4}. Denote by $\sysstate\invar_{k,T}$ the (unique) invariant state of $\L_{k,T}$, and $P_{k,T}^1$ the associated spectral projector of $\L_{k,T}$.
Let $\sysstate\init$ be the initial state of the system, and assume $(P_{0,T}^1 + Q_{0,T})\sysstate\init = \sysstate\init$. Further set 
\begin{equation}
\begin{gathered} \label{eq_defXD}
	\Udiff_{k,T}:=U_{k,T} \,\sysstate[k,T]\invar \otimes \envstate[k,T]\, U_{k,T}^* - \sysstate[k,T]\invar\otimes\envstate[k,T],\\
	D_{k,T} := \L_{k,T} ( \sysstate[k-1,T] - \sysstate[k,T]\invar) \otimes \envstate[k,T] - U_{k,T}\big( ( \sysstate[k-1,T] - \sysstate[k,T]\invar)\otimes \envstate[k,T] \big) U_{k,T}^*
\end{gathered}
\end{equation}
and $\delta_{k,T}=\frac1{16}\,\infspec\eta_{k,T}$ where $\eta_{k,T}={\sysstate\invar_{k,T}}\otimes \xi_{k,T}$.		
Then, there exist $T_0>0$ and $C^P>0$ depending only on $c_P$ defined by~\eqref{eq_defCPCPp}
, and~$C_{\eta_{k,T}}^P$ depending only on $c_P$, 
and on $\infspec\eta_{k,T}$, such that if
\begin{equation} \label{eq_conditionsProp1atomentropy}
	 T\geq \max\big(T_0,C^P(\delta_{k,T}(1-\ell))\inv\big),\qquad \|Q_0 \rho\init\|_1 \ell^k \leq \delta_{k,T}, \qquad \|\Udiff_{k,T}\|\leq \delta_{k,T},
\end{equation}
then with $F_{k,T} =F_{\eta_{k,T}}$ defined by \eqref{eq_defFeta},
\begin{equation}	
\label{eq:1slow}
 \big|\sigma_{k,T} -F_{k,T} (\Udiff_{k,T}-D_{k,T})\big|\leq C_{\eta_{k,T}}^P \big(  \|\Udiff_{k,T}\|_1 + \|Q_{0,T}\sysstate\init\|_1 \ell^k + (T(1-\ell))\inv \big)^3 .
 \end{equation}  
\end{proposition}
\begin{remarks} \, \hfill \vspace{-0.8em}
\begin{itemize}
\item The condition $(P_{0,T}^1 + Q_{0,T})\sysstate\init = \sysstate\init$ means that $\sysstate\init$ contains invariant parts and strictly contracting parts, but no part associated with other eigenvalues of $\L_0$ on the unit circle.
\item
If we assume $Q_{0,T} \sysstate\init = 0$ and $\delta_{k,T}>\delta>0$, then conditions \eqref{eq_conditionsProp1atomentropy} relate to $T$ and $\|X_{k,T}\|$ only.
\item The expression $D_{k,T}$ measures the discrepancy between the actual evolution when the adiabatic parameter is $T$ and the ideal adiabatic approximation where the instantaneous state is equal to the invariant state, as shown by the inequality $\|D_{k,T}\|_1 \leq 2 \|\sysstate[k-1,T] - \sysstate[k,T]\invar \|_1$. Equation\nobreakspace \textup {(\ref {eq:estimatestate})} below shows that, if $Q_{0,T}\rho\init=0$ then $\|D_{k,T}\|_1$ is of order$(T(1-\ell))\inv$.
\item The term $X_{k,T}$ is inherent to the non-equilibrium nature of the RIS dynamics. More precisely, it is zero only if the CPTP map $\L_{k,T}$ satisfies a detailed balance property, as shown in Lemma\nobreakspace \ref {lemma_caracX} below.
\item If the conditions hold for every $k$ and $T$, this Proposition allows us to compute an asymptotic rate of entropy production, as in Section\nobreakspace \ref {sec:scl}.  In that Section, we consider the small coupling version of this Proposition, in which the condition on $\|\Udiff_{k,T}\|$ holds for $\lambda$ small.
\end{itemize}
\end{remarks}
\begin{proof}
For the rest of this section, we drop the subscript $T$, and from now on, for any nonnegative quantity~$Z$, denote by $O_{\eta}^P(Z)$ (respectively $O^P(Z)$) any term that is bounded by a $C_\eta^P\, Z$, for~$Z$ small enough and $C_\eta^P$ depending only on a given state $\eta$ and the quantity $c_P$ 
associated with $s\mapsto P(s)$ (respectively, on $c_P$
). 
To estimate $\sigma_k$, we will estimate $\sysstate[k-1]$ in terms of~$\sysstate[k]\invar$, and apply~Proposition\nobreakspace \ref {prop:perturbentropy}. 

By Proposition\nobreakspace \ref {prop:normTospr}, $s\mapsto\L(s)$ satisfies \textup{H1}--\textup{H4}. Using Propositions\nobreakspace \ref {prop:Qkilled} and\nobreakspace  \ref {prop:general-ad}, there exists~$T_0$ depending only on $c_P$
such that, for $T\geq T_0$,
\begin{align*}	
\L_{k-1}\dotsm\L_{1} &=  A_{k-1} + \L_{k-1}^Q\dotsm \L_1^Q Q_0+ O^P\big((T(1-\ell))\inv\big).
\end{align*}
We apply this to our initial state $\sysstate\init$ to obtain,
\[ \sysstate[k-1] = A_{k-1} \sysstate\init + \L_{k-1}^Q\dotsm \L_1^Q Q_0 \sysstate\init+ O^P\big((T(1-\ell))\inv\big),\]
so that
\[\|\sysstate[k-1]-\sysstate[k]\invar\|_1 \leq \|A_{k-1} \sysstate\init-\sysstate[k-1]\invar\|_1 +\|\sysstate[k-1]\invar - \sysstate[k]\invar\|_1 + \|Q_0 \sysstate\init\|_1\ell^k + O^P\big((T(1-\ell))\inv\big).
\]

But $A_{k-1}\sysstate\init = A_{k-1}P_0^1 \sysstate\init = P_{k-1}^1 A_{k-1} \sysstate\init \in\ran P^1_{k-1} = \cc\,\sysstate[k-1]\invar$. Set $A_{k-1} \sysstate\init = \alpha \sysstate[k-1]\invar$. Then using that each $\L(s)$ is trace preserving,
\[1 = \tr(\sysstate[k-1]) = \tr(\L_{k-1}\dotsm \L_1 \sysstate\init) = \alpha + \tr(\L_{k-1}^Q \dotsm \L_1^Q \sysstate\init) + O^P\big((T(1-\ell))\inv\big),\]
so that $\|A_{k-1} \sysstate\init-\sysstate[k-1]\invar \|_1 = |1 - \alpha| \leq  \|Q_0 \sysstate\init\|_1\ell^k + O^P((T(1-\ell))\inv)$.
In addition,  $\|\sysstate[k-1]\invar - \sysstate[k]\invar\|_1 = O^P(T^{-1})$: indeed, we can write for $s, s_0\in [0,1]$ close enough, 
\begin{align*}
P^1(s)\sysstate[s_0]\invar=\gamma(s)\sysstate[s]\invar, \ \text{with} \ \gamma(s)=\tr P^1(s)\sysstate[s_0]\invar \neq 0,  \ \text{$\gamma$ is $C^2$, and} \ \gamma(s_0)=1. 
\end{align*}
Hence $\sysstate[s]\invar-\sysstate[s_0]\invar=({\gamma(s)\inv} P^1(s)-P^1(s_0))\sysstate[s_0]\invar = O^P(s-s_0)$ by Lemma\nobreakspace \ref {lem:projC2}.

Thus, the computation finally yields
\begin{align} \label{eq:estimatestate}
\|\sysstate[k]\invar - \sysstate[k-1]\|_1 \leq 2\|Q_0 \sysstate\init\|_1\ell^k +O^P\big((T(1-\ell))\inv\big).
\end{align}
Notice that $\|D_k\|_1 \leq 2\|\sysstate[k]\invar - \sysstate[k-1]\|_1$. Denote 
 \begin{align*}
D'_{k}= &\L_k(\sysstate[k-1])\otimes \envstate[k]- \sysstate[k]\invar\otimes \envstate[k]=\L_k(\sysstate[k-1]-\sysstate[k]\invar)\otimes \envstate[k],
\\
 D_{k}'' = &U_{k} (\sysstate_{k-1}\otimes \envstate[k]) U_k^* - \sysstate\invar_{k}\otimes \envstate[k] = U_k \big((\sysstate[k-1]- \sysstate[k]\invar) \otimes \envstate[k] \big) U_k^* + X_k.
 \end{align*}
We have $\tr D_{k}'=\tr D''_{k}=0$. Then $\sigma_k$ can be rewritten as
$S(\eta_k+ D''_{k}| \eta_k+ D'_{k})$, we have $D''_{k}-D'_{k} = X_k -D_k$, and
\begin{equation} \label{eq_boundDk}
\|D_k'\|_1 \leq  \|\sysstate[k-1] - \sysstate[k]\invar\|_1, \qquad \|D_k''\|_1 \leq  \|\sysstate[k-1] - \sysstate[k]\invar\|_1 + \|X_k\|_1.
\end{equation}
By relations \eqref{eq:estimatestate} and \eqref{eq_boundDk}, conditions \eqref{eq_conditionsProp1atomentropy} imply $\max(\|D_k'\|_1,\|D_k''\|_1) \leq \frac14\infspec\eta_k$ and we can apply~Proposition\nobreakspace \ref {prop:perturbentropy} to obtain
\begin{align}\label{fixedlambda}	
\sigma_k &= F_k(D_k -X_k) + O_{\eta_k}^P \big((\|Q_0 \sysstate\init\|_1 \ell^k +\|X_k\|_1+ (T(1-\ell))\inv )^3\big).
\end{align}
\qedhere
\end{proof}

The next result is about $m$-repeated interaction systems. Consistently with \eqref{eq_mrepeatedobs}, for  $1\leq j \leq m$ and $1\leq k \leq T$ we denote by e.g. $\sigma_{k,T}^{(j)}$, $\rho_k^{(j)}$, the quantities at time $(k-1)m+j$, i.e. after the interaction with the $j$-th copy of the $k$-th probe. For notational convenience, we extend this notation $\rho_k^{(j)}$ to $j=0$, so that $\rho_k^{(0)}=\rho_{k-1}^{(m)}$. Note that e.g. $\rho_{k,T}\invar{}^{(j)}$ and $X_{k,T}$ depend on $k$ and not on $j$.
\begin{corollary}[$m$-RIS]\label{cor:matomentropy}
Consider an interaction system satisfying assumption \textup{$m$-ADRIS}, where the CPTP map $\L(s)$ is irreducible for all $s\in[0,1]$ and satisfies \textup{wH4}, and $m$ is associated with $\ell\in {]\ell',1[}$ by Proposition\nobreakspace \ref {prop:normTospr}. Assume that the initial state $\sysstate\init$ satisfies $(P_{0,T}^1 + Q_{0,T})\sysstate\init = \sysstate\init$. Further set
\[D_{k,T}^{(j)}  := \L_{k,T} ( \sysstate[k,T]^{(j-1)} - \sysstate[k,T]\invar) \otimes \envstate[k,T] - U_{k,T}\big( ( \sysstate[k,T]^{(j-1)} - \sysstate[k,T]\invar)\otimes \envstate[k,T] \big) U_{k,T}^*\]
and $X_{k,T}$ as in \eqref{eq_defXD}. With $\delta_{k,T}$, $\eta_{k,T}$, $T_0$, $C^P$, $C_{\eta_{k,T}}^P$ as in Proposition\nobreakspace \ref {prop:1atomentropy} then under conditions \eqref{eq_conditionsProp1atomentropy} we have with $F_{k,T} =F_{\eta_{k,T}}$
\begin{equation} \label{eq:jslow}
 \big|\sigma_{k,T}^{(j)} - F_{k,T} (\Udiff_{k,T}-D_{k,T}^{(j)})\big|\leq C^P_{\eta_k} \big( ( \|\Udiff_{k,T}\|_1 + \|Q_{0,T}\sysstate\init\|_1 \ell^k + (T(1-\ell)\inv )^3 \big).
\end{equation} 
\end{corollary}
\begin{proof}
Again we drop the index $T$. We follow the proof of Proposition\nobreakspace \ref {prop:1atomentropy}, applied to $s\mapsto\L^m(s)$. Remark that the quantity $c^P$
associated with this map is the same as those associated with $s\mapsto \L(s)$. Therefore, for the same $T_0$ and $C^P$ we have for $T\geq T_0$
\begin{align*}
\|\sysstate[k]\invar - \sysstate[k]^{(0)}\|_1  \leq 2\|Q_0 \sysstate\init\|_1\ell^k +O^P\big((T(1-\ell))^{-1}\big).
\end{align*}
Since $\sysstate[k]\invar$ is invariant by $\L_k$, applying the contraction $\L_k^j$ for $1\leq j \leq m$ yields
\begin{align} \label{eq:estimatestatem}
\|\sysstate[k]\invar - \sysstate[k]^{(j)}\|_1 &= \| \L^j_k ( \sysstate[k]\invar - \sysstate[k-1]^{(0)})\|_1 
\leq 2\|Q_0 \sysstate\init\|_1\ell^k +O^P\big((T(1-\ell))^{-1}\big).
\end{align}
Since $\sigma_k^{(j)}=S\big(U_k(\rho_k^{(j-1)}\otimes \envstate[k])U_k^* | \L_k(\rho_k^{(j-1)})\otimes \envstate[k] \big)$ we can proceed as in the proof of~Proposition\nobreakspace \ref {prop:1atomentropy} to obtain Equation\nobreakspace \textup {(\ref {eq:jslow})}.
\end{proof}
\begin{remark}
Thus, an $m$-RIS simply has approximately an $m$-fold total increase in entropy as compared to the associated $1$-RIS.	
\end{remark}

We can now summarize the behaviour of $\sigma_T^\text{tot}$ in the $T\to\infty$ limit.
\begin{corollary} \label{cor:sup_inf_special_RIS_Xs}
Consider a repeated interaction system satisfying either

(i) \textup{ADRIS}, such that the reduced dynamics $\L(s)$ is irreducible for all $s\in [0,1]$ and satisfies \textup{H4}, or 

(ii) $m$-\textup{ADRIS}, where $m$ is associated with $\ell\in {]\ell',1[}$ by Proposition\nobreakspace \ref {prop:normTospr}, such that the reduced dynamics $\L(s)$ is irreducible for all $s\in [0,1]$ and satisfies \textup{wH4}.

Denote by $\sysstate[s]\invar$ the unique invariant state for $\L(s)$, and let
$X(s) = U(s)\big(\sysstate[s]\invar \otimes \envstate[s]\big) U(s)^* - \sysstate[s]\invar \otimes \envstate[s]$.

If $X(s) =0$ for all $s\in[0,1]$, and $\sysstate\init=\sysstate\invar_{0,T}$,  then  $\lim_{T\to \infty}\sigma_T^\text{tot}=0$.

If $X(s) =0$ for all $s\in[0,1]$, but $\sysstate\init\neq\sysstate\invar_{0,T}$,  then  $\lim_{T\to \infty}\sigma_T^\text{tot}<\infty$, but it is possibly non-zero.

If $\sup_{s\in[0,1]}\|X(s)\|_1 > 0$, then $\sigma^\text{tot}_T \to \infty$ in the limit $T\to\infty$.
\end{corollary}

\begin{proof}
First remark for any $k$ and $T$, 
	\[\delta_{k,T}\geq \frac1{16}\, \inf_s \infspec\sysstate[s]\invar \times \infspec \frac{\mathrm{e}^{-\beta(s)H_\env(s)}}{\tr(\mathrm{e}^{-\beta(s)H_\env(s))}} =: \delta >0.
\]
In addition, the constant $C_{\eta_k,T}^P$ in Proposition\nobreakspace \ref {prop:1atomentropy} is monotonically decreasing in $\infspec\eta_{k,T}$. By Proposition\nobreakspace \ref {prop:1atomentropy} (respectively Corollary\nobreakspace \ref {cor:matomentropy}), the inequality \eqref{eq:1slow} (respectively  \eqref{eq:jslow}) holds for any large $T$, if $\max(\|Q_0\rho\init\|_1 \ell^k,\|X_{k,T}\|)<\delta$, with $C_{\eta_k,T}^P$ replaced by a uniform $C$. If $X(s)=0$ for all $s$
 then we have
 \[|\sigma_{k,T}-F_{k,T}(D_{k,T})|\leq C\big(\|Q_0\rho\init\|_1\ell^k + (T(1-\ell))\inv)^3\big),\]
for all $k$ if $\|Q_0\rho\init\|=0$, and for $k$ large enough otherwise. By \eqref{eq_boundFeta} and \eqref{eq:estimatestate} we have for $T$ large enough that
\[F_{k,T}(D_{k,T})=O\big(\|Q_0\rho\init\|_1 \ell^{2k}+(T(1-\ell))^{-2}\big).\]
This proves our first two assumptions in the case (i) of an ADRIS. The same proof with e.g. $\sigma_{k,T}^{(j)}$ replacing $\sigma_{k,T}$  gives the case (ii) of an $m$-ADRIS. Now, if $\sup_{s\in[0,1]}\|X(s)\|_1 > 0$ then there exists an interval $I$ of diameter $\delta>0$, and some $x>0$, such that $\|X(s)\|_1>x$ for $s\in I$. By relation \eqref{eq:quantum_pinsker} one has $\sigma_{k,T}\geq \frac12\|D_{k,T}-X_{k,T}\|_1^2$, and by \eqref{eq:estimatestate}, for $T$ large enough one has $\|D_{k,T}-X_{k,T}\|_1>\frac12x$ if $k/T\in I$. Therefore, \[\sigma_T^\text{tot} := \sum_{k=1}^T \sigma_{k,T} \geq \sum_{k\,|\, k/T\in I} \frac18 x^2\underset{T\to\infty}\rightarrow +\infty.\]
Again the same proof holds in the case (ii) of an $m$-ADRIS. \qedhere
\end{proof}

Since the condition $X\equiv0$ plays an important role in the behaviour of the total entropy $\lim_{T\to\infty}\sigma_T^{tot}$, the following result is relevant. Its proof is obtained by straightforward manipulations.
\begin{lemma} \label{lemma_caracX}
	Under the assumptions of Proposition\nobreakspace \ref {prop:1atomentropy} or Corollary\nobreakspace \ref {cor:matomentropy}, and the notation of Corollary\nobreakspace \ref {cor:sup_inf_special_RIS_Xs}, the following are equivalent:
	\begin{enumerate}
	 	\item $X(s)=0$,
	 	\item there exists a self-adjoint operator $K_\sys(s)$ on $\H_\sys$ such that $[K_\sys(s)+h_{\env}(s),U(s)]=0$.
 	\end{enumerate}
	 Moreover, if either condition holds we have the detailed balance relation that for all $\rho\in \B(\H_\sys)$, 
	 	\begin{equation} \label{eq_qdb}
	 		(\sysstate[s]\invar{})^{+1/2} \L(s)^*\big((\sysstate[s]\invar{})^{-1/2}\rho\,(\sysstate[s]\invar{})^{-1/2}\big)\,(\sysstate[s]\invar{})^{+1/2} = \tr_\env\big( U(s)^*(\rho\otimes \envstate[s])\,U(s)\big).
	 	\end{equation}
\end{lemma}

\begin{remark}
	The relation \eqref{eq_qdb} requires some comments regarding the nature of both sides of the identity. The left-hand side $\widetilde \L(s)$ is known under various names, as the KMS-dual or standard time-reversal (see \cite{Cipriani,DerFru,FagUma,GolLin}, see \cite{Crooks} for a physical motivation). The right-hand side is the reduced evolution associated with $\envstate[s]$ and $U(s)^*$. Note that, considering $\theta$ to be complex conjugation on $\H_\sys\otimes \H_{\env}$ in a basis that is jointly diagonal for $K_\sys$, $h_\env(s)$, then $\theta=\theta_\sys\otimes \theta_\env$ and with $\Theta_\sys:X\mapsto \theta_\sys X\theta_\sys$ then the right-hand side of \eqref{eq_qdb} equals $\Theta_\sys\circ\L\circ \Theta_\sys$, so that the detailed balance condition \eqref{eq_qdb} takes the form $\widetilde \L_{k,T}=\Theta_\sys\circ\L\circ \Theta_\sys$, which is the SQDB-$\theta$ condition of \cite{FagUma}, and is simply called time-reversal invariance in \cite{JPW}. The condition $X(s)=0$ for all $s$ can then be related to non-adiabatic entropy production of \cite{HoroParr}.
\end{remark}

We apply the results above to the example of two-level systems coupled by their dipole, in the rotating wave approximation.

\subsection{Example: 2-level system with rotating wave approximation}\label{2lrw}
\label{sec:rotwave}

As a first example, we consider the example discussed at the end of Section\nobreakspace \ref {ssec:RISmathdesc} in the rotating wave approximation. Let $[0,1] \ni s \mapsto \beta(s) \in [0,\infty)$ be a $C^2$ function. The reduced dynamics on~$\sys$ then takes the form
\begin{align*}
	\L(s)\,\big(\eta\big) &= \tr_{\env} \big( e^{-\i\tau (\hsys + \henv + \lambda v_\text{RW})} (\eta \otimes \envstate_{\beta(s)}) e^{\i\tau (\hsys + \henv + \lambda v_\text{RW})} \big),
\end{align*}
where $\envstate_{\beta}$ denotes the Gibbs state of $\henv$ at inverse temperature $\beta$.

The spectral analysis of $\L(s)$ via its Kraus decomposition is given in \cite{BJM14}. It yields that the spectrum of $\L(s)$ is independent of $s \in [0,1]$, with $1$ a simple eigenvalue with eigenvector $\sysstate_{\beta^*(s) }$, the Gibbs state of $\hsys$ at inverse temperature $\beta^*(s) := \frac{E_0}E \beta(s)$. In particular, this eigenvector is a positive definite map and by Theorem 6.7 in \cite{wolftour} this implies that $\L(s)$ is irreducible. We have in addition that $\spr Q(s) \L(s)  = (1 - \tfrac{\lambda^2}{\Delta^2 + \lambda^2}\sin^2 \tfrac{\nu\tau}{2})^{1/2}$, where $\Delta := E - E_0$ and $\nu := \sqrt{\Delta^2 + \lambda^2}$, whenever $\nu\tau \notin 2\i \pi\Z$.

As can be seen from the Kraus decomposition,
\begin{align*}
	\L(s)\big(\eta\big) &= \big(P(s) + \sum_{i=1}^3 \theta^j Q^j(s)\big)\,\eta \\
		&= \tr(\eta) \sysstate_{\beta^*(s), E} + \theta^1 \tr(a^* \eta) a  \\
		& \qquad + \theta^2 \tr(a \eta) a^* + \theta^3 \tr((e^{-\beta(s) E_0}(\one-N) - N ) \eta) \frac{(\one-N) - N}{1+e^{-\beta E_0}},
\end{align*}
where $|\theta^j| \leq (1 - \tfrac{\lambda^2}{\Delta^2 + \lambda^2}\sin^2 \tfrac{\nu\tau}{2})^{1/2}$ (see \cite{BJM14}).

Also from this decomposition, $\L(s)$ does not in general satisfy \textup{H4}, but only \textup{wH4}. The $m$ for which Proposition\nobreakspace \ref {prop:normTospr} holds can be found explicitly by estimating $\|(\L(s)Q(s))^m\|$ uniformly in $s$ with the help of the above decomposition. The rescaling and discretization for the $m$-repeated RIS of this example is represented in Figure\nobreakspace \ref {fig:sampling}.

\begin{figure}[h]
	\begin{center}
\begin{tikzpicture}[scale=0.8]
\usetikzlibrary{decorations.pathreplacing}

\draw[->] (0,0) -- (0,4);
\draw[->] (0,0) -- (15,0);

\draw[dashed,lightgray] (6.2,-1.5) -- (5,-5.5);
\draw[dashed,lightgray] (14.5,-5.5) -- (11.2,-1.5);

\draw[gray] (0,1.5) .. controls (0.5,0.5) and (1,1) .. (1.5,1.5) .. controls (2,2) and (2.5,2.5) .. (3.5,2.5);
\draw[dashed,gray] (3.5,4) -- (3.5,0);
\draw[gray] (0,1.5) .. controls (1,0.5) and (2,1) .. (3,1.5) .. controls (4,2) and (5,2.5) .. (7,2.5);

\draw (0,1.5) .. controls (2,0.5) and (4,1) .. (6,1.5) .. controls (8,2) and (10,2.5) .. (14,2.5);
\draw[dashed] (14,4) -- (14,0);
\node at (3.5,-0.5) {1};

\node at (14,-0.5) {$Tm\tau$};
\node at (16,0) {\footnotesize time};

\draw  (7.2,2.2) rectangle (8.5,1.7);
\draw  (5,-1.5) rectangle (14.5,-5.5);

\draw (5,-5.2) -- (5.8,-5.2) -- (5.8,-4.5) -- (6.8,-4.5) -- (6.8,-4.5) -- (7.3,-4.5) -- (7.3,-4.5) -- (7.8,-4.5) -- (7.8,-3.9) -- (8.3,-3.9) -- (8.3,-3.9) -- (8.8,-3.9) -- (8.8,-3.9) -- (9.3,-3.9) -- (9.3,-3.9) -- (9.8,-3.9) -- (9.8,-3.2) -- (10.3,-3.2) -- (10.3,-3.2) -- (10.8,-3.2) -- (10.8,-3.2) -- (11.3,-3.2) -- (11.3,-3.2) -- (11.8,-3.2) -- (11.8,-2.6) -- (12.3,-2.6) -- (12.3,-2.6) -- (12.8,-2.6) -- (12.8,-2.6) -- (13.3,-2.6) -- (13.3,-2.6) -- (13.8,-2.6) -- (13.8,-2.1) -- (14.3,-2.1) -- (14.3,-2.1) -- (14.5,-2.1);

\draw[decorate, decoration={brace, amplitude=2pt}] (5.8,-3.7) -- (7.8,-3.7);
\node at (6.8,-3.4) {$m\tau$};
\node at (3,-4.5) {$\beta((k+1)/T)$};
\draw (4.9, -4.5) -- (5.1,-4.5);
\draw (7.2,1.7) -- (6.2,-1.5);

\draw (7.2,2.2) -- (5,-1.5);
\draw (8.5,1.7) -- (11.2,-1.5);
\draw (8.5,2.2) -- (14.5,-1.5);
\node at (0,4.5) {\footnotesize inv. temp.};
\node at (1.5,2.2) {$\color{gray}\beta(s)$};

\node at (5.2,-6) {$mk\tau$};
\node at (6.1,-6.5) {$mk\tau+\tau$};
\draw (5.8,-5.5) -- (5.8,-5.8) -- (5.7,-5.9);

\draw (6.3,-5.5) -- (6.3,-6.1) -- (6.1,-6.3);
\draw (7.8,-5.5) -- (7.8,-5.8) -- (7.9,-5.9);
\node at (9,-6) {$mk\tau+m\tau$};

\end{tikzpicture}
	\caption{For a curve $\beta(s)$ on $[0,1]$ 
	the $m$-repeated ADRIS is such that, between time $mk\tau$ and $mk\tau + m\tau$, the system $\sys$ interacts successively with $m$ atoms at inverse temperature $\beta{((k+1)/T)}$.}
	\label{fig:sampling}
	\end{center}
\end{figure}
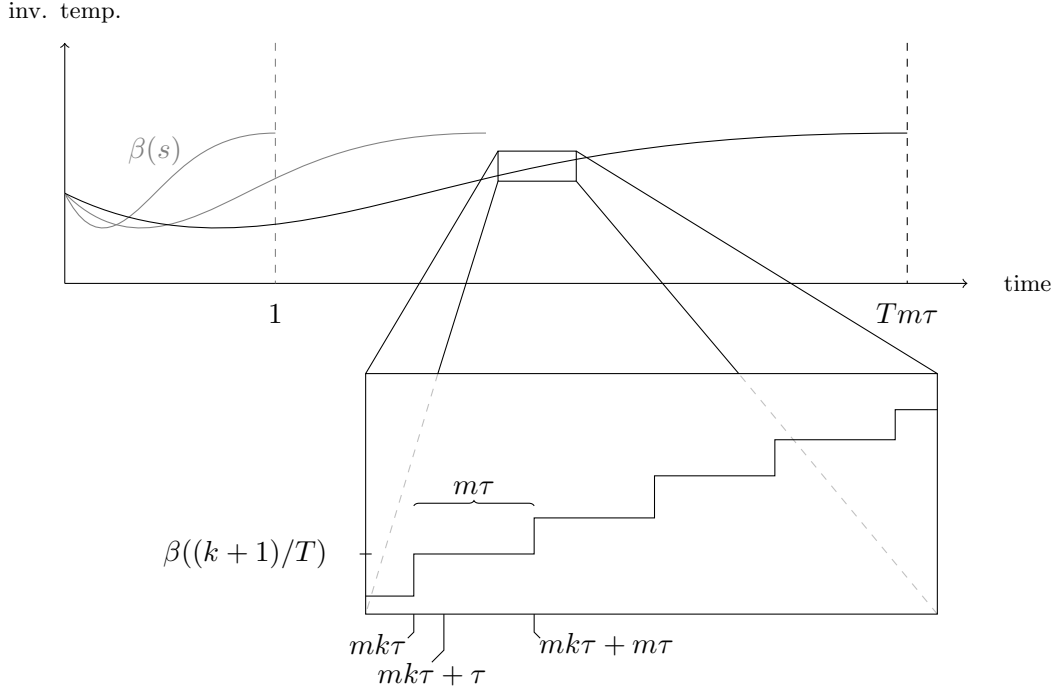

A key feature of this particular system is that the invariant state of the system at time $s$, i.e. $\sysstate_{\beta^*(s)}$, satisfies
\begin{align}
	\L(s)(\sysstate_{{\beta^*(s)}})\otimes \envstate_{\beta(s)} = \sysstate_{{\beta^*(s)}}\otimes \envstate_{\beta(s)} = U(s) \sysstate_{{\beta^*(s)}}\otimes \envstate_{\beta(s)} U(s)^*, \label{eq:rotwaveinvariance}
\end{align}
or in the above notation $\Udiff(s)= 0$ for all $s$. Hence, Corollary\nobreakspace \ref {cor:sup_inf_special_RIS_Xs} yields that, starting in the initial invariant state $\sysstate_0\invar = \sysstate_{\beta_0^*}$, we have $\sigma_T^\text{tot} := \sum_{k=1}^T \sum_{j=1}^m \sigma_{k,T}^{(j)} \underset{T\to\infty}\rightarrow 0$. That is, the total entropy production vanishes in the adiabatic limit.
Starting in another state we have $\lim_{T\to\infty}\sigma_T^\text{tot}<+\infty$.

\section{Small coupling limit}\label{sec:scl}
Corollary\nobreakspace \ref {cor:sup_inf_special_RIS_Xs} characterizes the entropy production  $\sigma_{k,T}$ of repeated interaction systems as a dichotomy of vanishing and divergent cases, when $T$ is independent of the other parameters. In the divergent case, we are naturally interested in obtaining an asymptotic rate of entropy production. The divergence is obtained by the lower bound \eqref{eq:quantum_pinsker}; in particular, because we do not necessarily have that $X_{k,T} < \delta_{k,T}$ to apply our perturbative result Proposition\nobreakspace \ref {prop:perturbentropy}, we do not have a leading order term for the entropy production.

Now we will consider the small coupling behaviour $|\lambda|\ll 1$ of Proposition\nobreakspace \ref {prop:1atomentropy}; this will allow computation of limits of the form $T\to\infty$ with $\lambda = \lambda(T)$. Since $X_{k,T} = O(\lambda)$ as will be shown in Lemma\nobreakspace \ref {lem:expandUwU}, by taking $\lambda \to 0$ as $T\to \infty$ in the right manner, we may achieve the assumption $X_{k,T} < \delta_{k,T}$. However, as the coupling parameter $\lambda \to 0$, the quantity $\ell(\lambda) = \sup_s \| \L(s) Q(s)\| \to 1$ generically, violating \textup{H4}. This has the interpretation that as the system and chain of probes decouple, more states become invariant; without well-separated eigenspaces of the reduced dynamics, we may not apply our adiabatic result Theorem\nobreakspace \ref {theo_adiabatic} to track the state of the system. We may address this by changing the physical setup to an $m(\lambda)$-RIS with $m$ increasing as $\lambda\to 0$. This has the effect of progressively slowing the joint evolution of system and chain of probes so that the invariant subspace of the reduced dynamics stays well-separated as $\lambda\to 0$.

Let us outline this procedure in more detail. First, we show that given an asymptotic expansion of the invariant state $\sysstate[k,T]\invar = (\sysstate[k,T]\invar)^{(0)} + \lambda (\sysstate[k,T]\invar)^{(1)} + O(\lambda^2)$, we may expand $X_{k,T}$ and compute the first order term, as shown in Lemma\nobreakspace \ref {lem:expandUwU}. Since the zeroth order term vanishes, in the regime $|\lambda|\ll 1$ we may achieve the assumption $\|X_{k,T}\|_1 < \delta_{k,T}$ in Proposition\nobreakspace \ref {prop:1atomentropy}. Next, we analyze the spectral properties of the reduced dynamics in Lemma\nobreakspace \ref {unifpert}, leading to Lemma\nobreakspace \ref {expp1} which shows that such an expansion of the invariant state exists. Further analysis in Lemma\nobreakspace \ref {regproj1} allows control over the $\lambda$-dependence of the quantity $c_P$ defined by Equation\nobreakspace \textup {(\ref {eq_defCPCPp})}. This allows us to use the adiabatic result Theorem\nobreakspace \ref {theo_adiabatic} to control the state of the system up to error $((1-\ell)T)\inv$, as formalized in Proposition\nobreakspace \ref {updateadiab}.
 Lemma\nobreakspace \ref {lemessy} shows that we choose a function $\lambda \mapsto m(\lambda)$ so that $\ell(\lambda) = \sup_s \|\L^{m(\lambda)}(s) Q(s)\|$ is bounded away from $1$. Together, these results yield Proposition\nobreakspace \ref {prop:1atomentropy_lambda}, the analog of Proposition\nobreakspace \ref {prop:1atomentropy}.
 
 From now on, we consider a fixed repeated interaction system satisfying assumption ADRIS or $m$-ADRIS, so that the functions $h_\env$, $\beta$ and $v$ are fixed. Therefore, we denote by $O(Z)$ any quantity that is bounded by $C Z$, for $C$ depending only on $h_\env$, $\beta$ and $v$.
 We first need some technical lemmas regarding the small $\lambda$ behaviour of the properties of the dynamics studied so far; we start with $\Udiff_{k,T}$. All proofs for this section will be given in Appendix\nobreakspace \ref {App:scl}.
\begin{lemma} \label{lem:expandUwU}
Assume $\sysstate[k,T]\invar$ admits the asymptotic expansion $\sysstate[k,T]\invar = (\sysstate[k,T]\invar)^{(0)} + \lambda (\sysstate[k,T]\invar)^{(1)} + O(\lambda^2)$ (uniformly in $k$). Then, the quantity $X_{k,T}$ defined in Proposition\nobreakspace \ref {prop:1atomentropy} satisfies
\begin{equation*}
\Udiff_{k,T} =\lambda M_{k,T} + O(\lambda^2)
\end{equation*}
with
\begin{align*}
M_{k,T}&:=U^{(0)} \big((\sysstate[k,T]\invar)^{(1)}\otimes \envstate[k,T]\big) (U^{(0)})^* - (\sysstate[k,T]\invar)^{(1)}\otimes \envstate[k,T] \\ &\qquad-\Big[(\sysstate[k,T]\invar)^{(0)}\otimes \envstate[k,T], \sum_{i} \pi_{i,k,T} v  \pi_{i,k,T} (-\i\tau) + \sum_{i \neq j} \pi_{i,k,T} v  \pi_{j,k,T} \Big( \frac{\exp(-\i\tau(E_{i,k,T}-E_{j,k,T}))-1}{E_{i,k,T}-E_{j,k,T}}\Big)\Big]
\end{align*}
where $\pi_{j,k,T}$, resp. $E_{j,k,T}$, are the spectral projectors, resp. eigenvalues, of $\ham_0 = \ham_\sys + \ham_{\env,k,T}$.
Note $M_{k,T}$ is traceless, self-adjoint, and depends on $\sysstate[k,T]\invar$, $\envstate[k,T]$, $v(k/T)$, and $\tau$, but is independent of~$\lambda$, and is bounded uniformly in $k$ and $T$. The error term is uniform in $k$ and $T$.
\end{lemma}
The proof is based on standard perturbative arguments. Next, we consider the $\lambda$ and $s$ dependence of the spectral data of 
\begin{equation}\label{deflls}
{\L}^\lambda(s)(\cdot )=\tr_{\env}(e^{ -\i\tau (\ham_0(s) + \lambda v(s))} (\cdot)\otimes\envstate(s)  e^{ \i\tau (\ham_0(s) + \lambda v(s))}),
\end{equation} 
where $\ham_0(s)=\ham_\sys+\ham_\env(s)$.
This is required in the assumptions of Lemma\nobreakspace \ref {lem:expandUwU} above, and to assess the dependence on $\lambda$ of  $c^P$ and $\ell$ that lead to our 
discrete non-unitary adiabatic theorem in Section\nobreakspace \ref {sec:Adiabatic}.

\begin{lemma}\label{unifpert} Assume $h_\env$, $\beta$ and $v$ are $C^2$ functions on $[0,1]$. The operator
${\L}^\lambda(s)$ defined by equation\nobreakspace \textup {(\ref {deflls})} is an entire function of $\lambda\in \C$ with $C^2$ coefficients in $s\in[0,1]$,
and there exists a numerical constant $C_0>1$ such that for $|\lambda|\tau\sup_{s\in[0,1]}\|v(s)\|$ small enough,
\[\| {\L}^\lambda(s)- {\L}^0(s)\|\leq C_0 |\lambda|\tau\sup_{s\in[0,1]}\|v(s)\|.\]
Moreover, for $z\not \in \sp ({\L}^0(s))$ and $\lambda$ s.t. $|\lambda|\tau\sup_{s\in[0,1]}\|v(s)\|\big\|({\L}^0(s)-z)^{-1}\big\|<1/C_0$,
\[\big\|({\L}^\lambda(s)-z)^{-1}- ({\L}^0(s)-z)^{-1}\big\|\leq C_0|\lambda|\,\big\| ({\L}^0(s)-z)^{-1}\big\|^2.\]
\end{lemma}
\begin{proof}
We write $\exp( -\i\tau (\ham_0(s) + \lambda v(s))$ as a Dyson series in interaction representation
\[e^{ -\i\tau (\ham_0(s) + \lambda v(s))}e^{ \i\tau \ham_0(s)}=\id +\sum_{n\geq 1}(-i\lambda)^n\int_{0\leq t_1\leq \dots\leq t_n\leq \tau}
v_I(s,t_n)v_I(s,t_{n-1})\dots v_I(s,t_1)\d t_n \dotsm \d t_1,\]
 where $v_I(s,t)=e^{-\i t \ham_0(s)}v(s)e^{ \i t \ham_0(s)}$ is $C^2$ in $(s,t)\in [0,1]\times [0,\tau]$. Each term of the series is $C^2$ in $s\in[0,1]$ and the norm of the $n^\text{th}$ term is bounded by $(|\lambda|\tau \sup_{s\in[0,1]}\|v(s)\|)^n/n!$. Inserting the series in (\ref{deflls}) and using our continuity assumptions, the coefficients of the resulting expansion are $C^2$ as well, so that uniform convergence yields the first statement.
 In particular, we have
 $e^{ -\i\tau (\ham_0(s) + \lambda v(s))}=e^{ -\i\tau \ham_0(s)}+\Delta U(\lambda, s)$, with $\|\Delta U(\lambda, s)\|\leq c_1 |\lambda|\tau \sup_{s\in[0,1]}\|v(s)\|$ for some numerical constant $c_1$, if $ |\lambda|\tau \sup_{s\in[0,1]}\|v(s)\|$ is small enough. Hence we can write for any~$\sysstate$ in~${\B}({\H}_{\sys})$
 \begin{align*}
& {\L}^\lambda(s)(\sysstate)-{\L}^0(s)(\sysstate )= \\ 
 & \ \ \ \tr _{\env} \big(\Delta U(\lambda, s)\sysstate\otimes\envstate(s)e^{ \i\tau \ham_0(s)} +e^{ -\i\tau \ham_0(s)}\sysstate\otimes\envstate(s)\Delta U(\lambda,s)+\Delta U(\lambda,s)\sysstate\otimes\envstate(s)\Delta U(\lambda,s)\big).
 \end{align*}
Taking the trace norm on ${\H}_{\sys}$ using $\| AB\|_1\leq \|A\| \|B\|_1$ and $\|  \tr _{\env} (A) \|_1\leq \|A\|_1$ for any $A,B\in{\cal B}({\H}_{\sys}\otimes{\H}_{\env})$  (see \cite{Ras}), and $\|\envstate(s)\|_1=1$, we get 
\[
\|{\L}^\lambda(s)(\sysstate)-{\L}^0(s)(\sysstate )\|_1\leq (2 \|\Delta U(\lambda, s)\|+\|\Delta U(\lambda, s)\|^2)\|\sysstate\|_1,
\]
from which the second statement follows. The last statement is then a direct consequence of the second resolvent identity. \qedhere

When $\lambda=0$, the map ${\L}^0(s)(\cdot) =e^{ -\i\tau \ham_{\sys}}(\cdot) e^{ \i\tau \ham_{\sys}}$ is independent of $s$ and has spectrum
$\sp ( {\L}^0(s))=\{e^{\i\tau (E_j-E_k)}\}_{E_j,E_k\in\sp (\ham_{\sys})}$. Thus $1$ is an eigenvalue of multiplicity at least $\dim \H_\sys$ and the other eigenvalues lie on the unit circle. We shall assume the following genericity condition which forbids accidental degeneracies:
\smallskip

\paragraph{GEN} The spectrum of ${\L}^0(s)$ consists of $\dim \H_\sys(\dim \H_\sys-1)$ simple eigenvalues different from~1, and 1 which is $\dim \H_\sys$-fold degenerate. Furthermore, hypotheses \textup{H1}, \textup{H2}, \textup{H3} and \textup{H4} hold for all $\lambda\in \R^*$  small enough, uniformly in $s\in[0,1]$.
\smallskip

\begin{remark}
This implies that the potential $v(s)$ couples the system effectively for all $s\in[0,1]$. In particular, $v(s)$ cannot be equal to zero for some values of $s$. Actually, the stronger Fermi Golden assumption means that for small but non zero $\lambda$, the peripheral spectrum is reduced to the eigenvalue 1. 
\end{remark}

Let us denote by $\tilde e^j(\lambda,s)$, $j=1,\cdots, \dim \H_\sys(\dim \H_\sys-1)$ the eigenvalues of ${\L}^\lambda(s)$ such that $\tilde e^j(0,s)\neq 1$ and 
by $\tilde e_1^j(\lambda,s)$,  $j=1,\cdots, J\leq \dim \H_\sys$, the distinct eigenvalues such that $\tilde e_1^j(0,s)=1$. One has  $J<\dim \H_\sys$ in case of degeneracy. The constant eigenvalue has index 1, i.e. $\tilde e_1^1(\lambda,s)\equiv1$. 

Let $\epsilon_0>0$ be such that the $\dim \H_\sys(\dim \H_\sys-1)+1$ balls of radius $\epsilon_0$ centred at $\sp ( {\L}^0(s))$ are distinct and a distance at least $\epsilon_0$ apart.  By Lemma \ref{unifpert}, there exists $l_0$  small enough such that  for all  $\lambda\in D({l_0})=\{ z\in\C \ | \ |z|<l_0\}$, $|\tilde e_1^j(\lambda,s)-1|<\epsilon_0$ and $|\tilde e^j(\lambda,s)-\tilde e^j(0,s)|<\epsilon_0$,  for all $j=1,\cdots, J$, uniformly $s\in [0,1]$. Analytic perturbation theory then implies that 
the eigenvalues $\tilde e^j(\lambda,s)$ and associated eigenprojectors $\tilde P^j(\lambda,s)$ are analytic in $\lambda\in D({l_0})$, for any fixed $s\in [0,1]$,  and all $j=1,\cdots, \dim \H_\sys$. Because of the splitting of the eigenvalues $\tilde e^j_1(\lambda,s)$ taking place at $\lambda=0$ where $\tilde e_1^j(0,s)=1$ for $j=1,\cdots, J\leq \dim \H_\sys$, these eigenvalues can have 
branching points at $\lambda=0$ and the corresponding eigenprojectors $\tilde P_1^j(\lambda,s)$ and eigennilpotents $\tilde N_1^j(\lambda,s)$ can even diverge there. 

We first prove that this does not occur for $\tilde P_1^1(\lambda,s)$ associated to $\tilde e_1^1(\lambda,s)\equiv1$:

\begin{lemma}\label{expp1} Let $s\in [0,1]$ be fixed. The eigenprojector $\tilde P_1^1(\lambda,s)$, for $\lambda\neq 0$ admits an analytic extension in a neighbourhood of the origin. Moreover,
the unique invariant state $\sysstate_s\invar(\lambda)$ is analytic in $\lambda$ in a neighbourhood of the origin.
\end{lemma}

\begin{corollary}\label{2dim}
Under the same hypotheses, if $\dim \H_\sys=2$, all eigenvalues and eigenprojectors of $\L^\lambda$ are analytic in $\lambda$ in a neighbourhood of the origin. Moreover, $\tilde e_1^2(\lambda)=1+O(\lambda^2)$.
\end{corollary}
\begin{proof} We only need to consider $\tilde e_1^2(\lambda)$ and $\tilde P_1^2(\lambda)$.
The assumption $\dim \H_\sys=2$ implies that for $\lambda\in D_0\setminus\{0\}$, $\tilde e_1^2(\lambda)\neq 1$ is non-degenerate. Since the 1-group $\tilde e_1^1(\lambda)+\tilde e_1^2(\lambda)=1+\tilde e_1^2(\lambda)$  and  the total projection on the 1-group   $\tilde P_1^1(\lambda)+\tilde P_1^2(\lambda)$ are analytic in $D_0$ (see \cite{Kato} Ch.II.\S 2) the above lemma implies the first statement. Finally, the property $\sp (\L)=\overline{\sp (\L)}$ implies that $\tilde e_1^2(\lambda)$ is real analytic in $\lambda$, and the assumption $|\tilde e_1^2(\lambda)|<1$ for $|\lambda|$ small and real imposes the term of first order in $\lambda$ to vanish.
\end{proof}

Finally, we address the dependence on $\lambda\in \R$ of the different estimates in Section\nobreakspace \ref {sec:Adiabatic} involved in the proof of our non-unitary adiabatic theorem. 

Let us relabel the eigenprojectors according to the convention used in \textup{H1} to \textup{H4}: $P^j$'s are associated to eigenvalues of modulus one, $P^1$ is associated with 1, and $Q^j$'s to eigenvalues of modulus strictly inferior to one when $0<|\lambda|<l_0$ is real and small, where $l_0$ is related to $\epsilon_0$ in the construction before Lemma\nobreakspace \ref {expp1}.

\begin{lemma}\label{regproj}\label{regproj1}
Let $\lambda\in \R$ such that $0<|\lambda|<l_0$ and $P^j(\lambda,s)$, $j\geq 1$, be associated to a peripheral eigenvalue. There exists $0<\tilde l_0\leq l_0$ and a constant $\tilde C_0$, such that for any $s\in [0,1]$, and any real $\lambda$ such that $0<|\lambda|<\tilde l_0$,
\[
\max_{k=0,1,2}\| \partial_s^k P^j(\lambda,s)\|\leq \tilde C_0, \ \ \forall j\geq 1.
\]
\end{lemma}

\begin{remarks}\, \hfill \vspace{-0.8em}
\begin{itemize} 
\item We get  under the same hypotheses that  $\max_{k=0,1,2}|\partial_s^ke^j(\lambda,s)|\leq \tilde C_0$, $j\geq 1$. It is enough to notice that $e^j(\lambda,s)=\L^\lambda(s)P^j(\lambda, s)$. 
\item Consequently, the constant $c^P$ defined in (\ref{eq_defCPCPp}) appearing in the estimates of Section\nobreakspace \ref {sec:Adiabatic} is uniform in $\lambda\in\R^*$, $|\lambda|$ small enough and the constant $N_{\text{max}}$ is uniform in $\lambda$ thanks to (\ref{eq_unifboundNs}).
\end{itemize}
\end{remarks}

Summarizing, we thus get the following control on the $\lambda$ dependence of our adiabatic theorem: 
\begin{proposition}\label{updateadiab}
Consider (\ref{deflls}) and assume GEN and ADRIS. Then, there exists $\tilde C$ such that for $\lambda\in \R^*$, $|\lambda|$ and $1/T$  small enough, independently of each other,
\[
\| \L_{n,T} \L_{n-1,T}\dotsm  \L_{1,T}  -  A_{n,T}-\L^Q_{n,T} \L^Q_{n-1,T}\dotsm  \L^Q_{1,T} Q_0\|\leq \frac{\tilde C}{(1-\ell(\lambda))T}, 
\]
where $\L_{k,T}=\L^\lambda(k/T)$,   $\L^Q_{k,T}=\L^\lambda(k/T)Q(\lambda, k/T)$. The $\lambda-$dependent operator $A_{n,T}$ is defined in Proposition\nobreakspace \ref {prop:general-ad}, and is bounded  uniformly in $\lambda$ and
\[
\ell(\lambda)=\sup_{s\in[0,1]}\|\L^{\lambda}(s)Q(\lambda,s)\|<1
\] is defined in \textup{H4}. 
\end{proposition}

\begin{remark}

When $\lambda\rightarrow 0$, $\ell(\lambda)$ approaches or exceeds 1. Indeed, in general 
$\|\L^{\lambda}(s)Q(\lambda,s)\|\geq \spr (\L^{\lambda}(s)Q(\lambda,s))$, and this spectral radius equals  $1$ at $\lambda=0$. 

\end{remark}

Therefore, we need to resort to \textup{wH4} and consider ${\L^{\lambda}(s)}^m Q(\lambda,s)$  instead of $\L^{\lambda}(s)Q(\lambda,s)$. Then,  the exponent $m=m(\lambda)$ increases as $\lambda\rightarrow 0$ in order to keep the norm of the latter operator smaller than one. This  amounts, in a sense, to changing the adiabatic time scale $T$ as a function of $\lambda$; see $m$-ADRIS. 

\begin{lemma}\label{lemessy}
Assume GEN with \textup{wH4} instead of \textup{H4}.  For all 
 $0<G<1$, there exists $m(\lambda)\in \N$ such that for any $\lambda\in\R^*$ small enough,
\[
\ell(\lambda)=\sup_{s\in[0,1]}\|{\L^\lambda(s)}^{m(\lambda)}Q(\lambda,s)\|\leq1-G.
\]
We can take $m(\lambda)\geq M_0\frac{\ln(1/|\lambda|)}{|\lambda|^r}$, where $M_0>0$ is a constant and $r>0$ stems from the estimate
\begin{align}\label{estspr}
\sup_{s\in[0,1]}\spr (\L^\lambda(s)Q(\lambda,s))\leq 1-S_0|\lambda|^r, \ \mbox{for some $S_0>0$}.
\end{align}
If $\dim \H_\sys =2$, $\ell(\lambda)\leq1-G$ for $m(\lambda)\geq M/\lambda^2$, for some $M>0$, for $\lambda\in\R^*$ small enough, generically.
\end{lemma}

\begin{remarks} \, \hfill \vspace{-0.8em}
\begin{itemize}
\item Corollary\nobreakspace \ref {2dim} ensures that all spectral data are analytic when $\dim \H_\sys=2$. By generic, we mean that coefficients of power series that need not be zero are, in fact, non zero. See the examples below.
\item Allowing $G$ to tend to one as $\lambda\rightarrow 0$  does not prevent $m(\lambda)$ to diverge in this limit. Nevertheless, the gain of considering ${\L^\lambda(s)}^{m(\lambda)}$ with \textup{wH4}, is that $\ell(\lambda)$ becomes independent of $\lambda$ and bounded away from one. Thus, the error terms in the adiabatic approximation become $\lambda$ independent, see Proposition\nobreakspace \ref {updateadiab}.
\end{itemize}
\end{remarks}

\begin{remark}
While $\|Q(\lambda,s)\|\leq 2$, for all $\lambda$ real and small, 
we don't have an a priori control on the individual projectors and eigennilpotents of $Q\L$.
\end{remark}

\begin{proposition}[$m$-RIS]\label{prop:1atomentropy_lambda}
Consider a repeated interaction system as described in Section\nobreakspace \ref {sec:RIS}, satisfying assumption $m$\textup{ADRIS}, and such that the induced CPTP map $\L^\lambda(s)$ is ergodic with simple eigenvalue 1 for all $s\in[0,1]$ and satisfies the assumptions of Lemma\nobreakspace \ref {lemessy}.
Denote by $\sysstate\invar_{k,T}(\lambda)$ the (unique) invariant state of $\L^{\lambda}_{k,T}$, and $P_{k,T}^1(\lambda)$ the associated spectral projector of $\L^{\lambda}_{k,T}$. Assume this state is faithful up to  $\lambda=0$.
Let $\sysstate\init$ be the initial state of the system, such that $\big(P_{0,T}^1(\lambda) + Q_{0,T}(\lambda)\big)\sysstate\init = \sysstate\init$. 
With the notations of Corollary\nobreakspace \ref {cor:matomentropy}, for $T$ large enough and $\lambda$ small enough, for $k\leq T$ large enough, we have 
$\sysstate[k,T]\invar (\lambda)= (\sysstate[k,T]\invar)^{(0)} + \lambda (\sysstate[k,T]\invar)^{(1)} + O(\lambda_k^2)$ and, 
with $M_{k,T}$ defined in Lemma\nobreakspace \ref {lem:expandUwU}, for all $1\leq j\leq m(\lambda)$,
\begin{align} \label{eq:1slowandweak}	
\sigma^{(j)}_{k,T} &= \lambda^2 F_{k,T}^{(0)}(M_{k,T},M_{k,T}) + F_{k,T}^{(0)}(D^{(j)}_{k,T},D^{(j)}_{k,T}) - \lambda F_{k,T}^{(0)}(D^{(j)}_{k,T},M_{k,T})\\
&\qquad - \lambda F_{k,T}^{(0)}(M_{k,T},D^{(j)}_{k,T}) +  O(\{\lambda + \|Q_{0,T} \sysstate\init\| \ell(\lambda)^k + T\inv \}^3), \nonumber
\end{align}
where $F_{k,T} ^{(0)}(\cdot, \cdot)=F_{(\sysstate\invar_{k,T})^{(0)}}(\cdot, \cdot)$.
\end{proposition}

\end{proof}

\begin{proof}
We drop the $T$ indices in this proof. We start form Corollary\nobreakspace \ref {cor:matomentropy} and seek its small $\lambda\in \R^*$ behaviour. First, Lemma\nobreakspace \ref {lemessy} ensures that considering ${\L^\lambda(s)}^{m(\lambda)}$, the adiabatic parts of the estimate are uniform in $\lambda$. Then we recall the estimate
\[
\|D^{(j)}_{k}\|_1\leq  2 \| \sysstate[k]\invar - \sysstate[k-1]\|_1 \leq 2\|Q_0 \sysstate\init\|_1\ell^k +O(T^{-1}).
\] that holds for all $1\leq j\leq m(\lambda)$.
For $\lambda$ small enough we can use the expansion of Lemma\nobreakspace \ref {lem:expandUwU}, as Lemma\nobreakspace \ref {expp1} states that $\sysstate_k\invar(\lambda)$ admits expansion in powers of $\lambda$ to all orders. Hence we can write $X_k = \lambda M_k + O(\lambda^2)$. Also, the perturbation formula (\ref{perteta}) and Lemma\nobreakspace \ref {expp1} show we can replace $F_k(\cdot, \cdot)=F_{\sysstate_k\invar}(\cdot, \cdot)$ by $F_{k}^{(0)}(\cdot, \cdot)=F_{{\sysstate_k\invar}(0)}(\cdot, \cdot)$ in (\ref{fixedlambda}), adding a term of order $\lambda$ in the curly bracket of the error term. 
Together with Lemma\nobreakspace \ref {lem:expandUwU}, this eventually yields Equation\nobreakspace \textup {(\ref {eq:1slowandweak})}.
\qedhere

\end{proof}

	Proposition\nobreakspace \ref {prop:1atomentropy_lambda} is the main result of this Section:  it provides an explicit formula to approximate the entropy production of an RIS at each step, as a function of all of the parameters, any of which could change between steps. Let us apply this result. Consider an $m(\lambda)$-RIS with the same assumptions as~Proposition\nobreakspace \ref {prop:1atomentropy}, and $m(\lambda)$ chosen by Lemma\nobreakspace \ref {lemessy}. Assume $Q_0\sysstate\init=0$ so that $D_{k,T}^{(j)} = O(T^{-1})$ by \eqref{eq_boundDk}. Since in addition, $\inf\sp\,(\rho_{s}^{\mathrm{inv}})^{(0)}>0$, a continuity argument shows the existence of a lower bound for the spectrum of $(\rho_{k,T}^{\mathrm{inv}})^{(0)}$, uniform in $k,T$. By Equation\nobreakspace \textup {(\ref {eq_boundFeta})} and Proposition\nobreakspace \ref {prop:1atomentropy}, one has for all $1\leq j \leq m(\lambda)$,
\begin{align*}	
\sigma_{k,T}^{(j)} &= \lambda^2 F^{(0)}_{k,T}(M_{k,T}) + O(1/T^2) + O(\lambda/T) +O(\lambda^3).
\end{align*}\
Then, summing over all interaction steps, 
\begin{equation}	
\sigma_{\lambda,T}^\text{tot} 
= m(\lambda)\Big(\lambda^2 TF_0 + O(1/T) + O(\lambda)+O(T\lambda^3)\Big) \label{eq_sigmatotmRIS}
\end{equation}
where
\[
0\leq \lim_{T\to\infty}\frac{1}{T}\sum_{k=1}^T  F^{(0)}_{k,T}(M_{k,T}) =: F_0  < \infty.
\]
Indeed,  $\L_{k,T}$ is obtained by sampling a $C^2$ function $\L(s)$. As noted earlier, by the spectral assumptions on $\L(s)$, we also obtain a $C^2$ function which is sampled to obtain the invariant state: $\sysstate[k,T]\invar= \sysstate\invar(s)$ for $s=k/T$ and $ F^{(0)}_{k,T}(\cdot, \cdot) = F_{ \sysstate\invar(s)^{(0)}}(\cdot, \cdot)$ for $s=k/T$. The same holds for $\ham_{{\env},k,T}=\ham_{\env}(k/T)$, $v_{k,T}=v(k/T)$. Thus the dependence of $k$ and $T$ of $M_{k,T}$ is of the form $M(s)$ for $s=k/T$ as well , for the $C^2$ function 
\begin{align*}	
M(s) &= \Big[(\sysstate\invar(s))^{(0)}\otimes \envstate(s), \sum_{i} \pi_i(s) v  \pi_i(s) (-\i\tau) + \sum_{i \neq j} \pi_i(s) v  \pi_j(s) \left( \frac{\exp(-\i\tau(E_i(s)-E_j(s)))-1}{E_i(s)-E_j(s)}\right)\Big]\\
&\qquad+ U^{(0)}(s) (\sysstate\invar(s))^{(1)}\otimes \envstate(s) (U^{(0)}(s))^*-(\sysstate\invar(s))^{(1)}\otimes \envstate(s).
\end{align*}
Then,
 \begin{align*}	
\lim_{T\to\infty}\sum_{k=1}^T \frac{1}{T} F^{(0)}_{k,T}(M_{k,T}) &= \lim_{T\to\infty}\sum_{k=1}^T F_{ \sysstate\invar(k/T)^{(0)} \otimes \envstate(k/T)}(M(K/T)) (k/T - (k-1)/T)\\
&= \int_0^1 F_{ \sysstate\invar(s)^{(0)}\otimes \envstate(s)}(M(s))\d s = F_0.
\end{align*}

\begin{remarks} \, \hfill \vspace{-0.8em}

\begin{itemize}
\item For an $m(\lambda)$-RIS as above, with $\lambda>0$ small but finite, the adiabatic limit $T\rightarrow \infty$ yields diverging entropy production if $F_0>0$, with a rate of entropy production in the adiabatic limit is given by $\lim_{T\to\infty} \frac{1}{T}\sigma_{\lambda,T}^\text{tot} = m(\lambda) \lambda^2 F_0 + O(\lambda^3)$.
\item  By the same arguments as in Corollary\nobreakspace \ref {cor:sup_inf_special_RIS_Xs}, we recover its analog:  in  the limit $T\to\infty$, $\lambda\to 0$, $\lambda T> \textit{constant}$, we recover  vanishing entropy production in the case $M_{k,T} \equiv 0$ and divergent entropy production in the case $\sup_{s\in [0,1]} \|M(s)\|_1 > 0$.
\item In particular, when $\dim \H_\sys =2$ we may take $m(\lambda)=[M_0/\lambda^2]$ to find
\[
\sigma_{\lambda,T}^\text{tot} = TM_0 F_0\big(1 + O(1/(T^2\lambda^2) + O(1/T\lambda)+O(\lambda^3)\big),
\]
for $F_0>0$. Thus $\sigma^\text{tot}_{\lambda,T}$ diverges like $T$ in the limit $T\to\infty$, $\lambda\to 0$, $\lambda T> \textit{constant}$, with an asymptotic rate of  $M_0 F_0 + O(\lambda^3)$.
\end{itemize}
\end{remarks}

We saw in Section\nobreakspace \ref {sec:rotwave} that, for the simple $2\times 2$ system with the rotating wave interaction from \cite{BJM14}, we do have relation $U_k \sysstate[k]\invar\otimes \envstate[k] U_k^* = \sysstate[k]\invar\otimes \envstate[k]$. 
We consider below the $2\times 2$ system with the full dipole interaction, and compute $F_k^{(0)}(M_k)>0$ for each step $k$.

\subsection{Example: 2-level system with full dipole interaction} \label{sec:fulldip}

Now we consider the same setup as~Section\nobreakspace \ref {sec:rotwave}, but with the full dipole interaction $v_{\text{FD}}= \frac{u_1}{2}( (a+a^*)\otimes (b+b^*))$. We choose the parameters $\{\tau, E_0, E\}$ so that  $\sp (\L^{0})$ is simple, except for the eigenvalue $1$ that is twice degenerate. By perturbation theory, for any $\beta(s)>0$, and any $\lambda>0$ small, hypothesis GEN holds.  
Hence the resulting reduced dynamics operator $\L^\lambda(s)$ is ergodic and satisfies \textup{wH4} (the eigenvalues are again independent of $\beta(s)$); see~Figure\nobreakspace \ref {fig:eigfull} for a particular choice of parameters.
\begin{figure}
\begin{center}
\begin{tikzpicture}[scale=1]
\draw[->] (-1.5,0) -- (1.5,0);
\node at (0.25,1.5) {$\i\R$};
\draw[->] (0,-1.5) -- (0,1.5);
\node at (1.5,0.25) {$\R$};
\draw (0,0) circle (1);
\draw[fill=red] (0.685604,0.289887) circle (0.025);
\draw[fill=red] (0.685604, -0.289887) circle (0.025);
\draw[fill=red] (0.554087, 0) circle (0.025);
\draw[fill=red] (1, 0) circle (0.025);
\node at (1.1,-0.15) {$1$};
\end{tikzpicture}
\end{center}
\caption{Numerically obtained eigenvalues of $\L$ with $\lambda=2$, $\tau=0.5$, $E_0=0.8$, and $E=0.9$. The eigenvalues of $\L$ are independent of $\beta$.}
\label{fig:eigfull}
\end{figure}
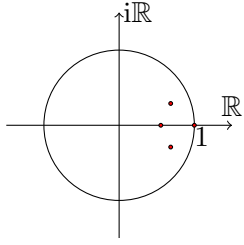
 The (unique) invariant state is
\begin{align*}	
\sysstate\invar_{k,T} &=\left(
\begin{array}{cc}
 \frac{e^{\beta_{k,T}  E_0} (1-\cos (\nu  \tau )) \eta ^2+\nu ^2
   (1-\cos (\eta  \tau ))}{\left(1+e^{\beta_{k,T}  E_0}\right)
   \left((1-\cos (\nu  \tau )) \eta ^2+\nu ^2 (1-\cos
   (\eta  \tau ))\right)} & 0 \\
 0 & \frac{(1-\cos (\nu  \tau )) \eta ^2+e^{\beta_{k,T}  E_0}
   \nu ^2 (1-\cos (\eta  \tau ))}{\left(1+e^{\beta_{k,T} 
   E_0}\right) \left((1-\cos (\nu  \tau )) \eta ^2+\nu ^2
   (1-\cos (\eta  \tau ))\right)} \\
\end{array}
\right)
\end{align*}
where $\nu=\sqrt{\left(E_0-E\right){}^2+\lambda ^2}$ and $\eta= \sqrt{\left(E+E_0\right){}^2+\lambda ^2}$. All matrices in this section are written in the basis given in the example at the end of~Section\nobreakspace \ref {ssec:RISmathdesc} (the eigenbasis of the unperturbed Hamiltonian $\ham_0 = \ham_\sys + \ham_\env$). We note that $\sysstate\invar_{k,T}$ has no first-order dependence on $\lambda$; the first correction is second order. We assume the system starts in the invariant state: $\sysstate\init = \sysstate\invar_0$. In the language of the proof of Proposition\nobreakspace \ref {prop:1atomentropy}, $D_{k,T}^{(j)} = O(T^{-1})$ and the quantity $X_{k,T}$ does not vanish. More precisely,
\begin{align*}
M_{k,T} &= A_{k,T} 
\begin{psmallmatrix}
0 & 0 & 0 & - \left(e^{i
   \tau  \eta _0}-1\right) \sin^2 (\nu_0 \tau/2) \eta _0
   \\
   0&0 & -\left(e^{i
   \tau  \nu _0}-1\right)   \sin^2 (\eta_0 \tau/2) \nu _0 & 0
   \\
    0 & 
  e^{-i\nu_0 \tau}\left(e^{i
   \tau  \nu _0}-1\right)   \sin^2 (\eta_0 \tau/2) \nu _0 & 0 & 0 
\\
  e^{-i\eta_0 \tau}\left(e^{i
   \tau  \eta _0}-1\right) \sin^2 (\nu_0 \tau/2) \eta _0 & 0 & 0 &0
   \end{psmallmatrix},
\end{align*}
where
\begin{align*}
A_{k,T} &=  \frac{\frac{1}{2}\tanh \left( \frac{\beta_{k,T} E_0}{2} \right)u_1}{{2 E_0 E\sin (E_0 \tau ) \sin (E \tau )-\left(E_0^2+E^2\right) (1-\cos (E_0 \tau ) \cos (E \tau))}}
\end{align*}
and $\nu_0 = |E-E_0|$ and $\eta_0 = |E+E_0|$. Therefore, we can use the fact that $[\sysstate\invar_{k,T} \otimes \envstate[k,T]\init, \ham_0] = 0$ to compute $F_{k,T}(M_{k,T})$ and conclude for small $\lambda$
\begin{align*}
	\sigma^{(j)}_{k,T} &= \lambda^2 \gamma \frac{\beta_{k,T} E_0}{2} \tanh \left( \frac{\beta_{k,T} E_0}{2} \right)  +  O(\lambda T^{-1}) + O(T^{-2})  + O(\lambda^3)
\end{align*}
where
\[
\gamma:=\begin{cases}
 \frac{ u_1^2(\cos (E_0\tau)-\cos (E \tau ))^2}{2 E_0 E\sin (E_0 \tau ) \sin (E \tau )-\left(E_0^2+E^2\right) (1-\cos (E_0 \tau ) \cos (E \tau))} & E\neq E_0\\
    \frac{2u_1^2\tau^2 \sin^2(E_0 \tau)}{1 + 2 E_0^2 \tau^2 - \cos(2 E_0 \tau)} & E=E_0.
  \end{cases}
\]
Hence, with $m(\lambda) = [M_0/\lambda^2]$  and  for small, constant $\lambda$, the total entropy production diverges as~$T$ in the adiabatic limit (for generic values of $E$, $E_0$ and $\tau$). In particular, we have the asymptotic rate
\begin{equation}
\lim_{T\to\infty}\frac{1}{T}\sigma_{T}^\text{tot} = \lambda^2 M_0\mathcal{E}\int_0^1  \frac{\beta(s) E_0}{2} \tanh \left( \frac{\beta(s) E_0}{2} \right)\d s + O(\lambda^3). \label{eq:mean}
\end{equation}

\section{Concluding remarks} \label{sec_concludingremarks}

Repeated Interaction Systems (RIS),  consisting of a small system of interest interacting with a chain of thermal probes, are naturally described by a discrete time evolution due to the instantaneous swapping of the probes. We consider the adiabatic regime where the probe parameters vary slowly from one probe to the next.

With our adiabatic theorem for discrete non-unitary evolution, Theorem\nobreakspace \ref {theo_adiabatic}, and a perturbative formula for relative entropy, Proposition\nobreakspace \ref {prop:perturbentropy}, we have characterized entropy production and saturation of the Landauer bound for finite-dimensional RIS in the adiabatic regime; see Corollary\nobreakspace \ref {cor:sup_inf_special_RIS_Xs}. The total entropy production of the RIS diverges if and only if  the  invariant state of the reduced dynamics at each step fails to be invariant under the full dynamics at each step\,---\,a condition that we linked to the physical notion of detailed balance in Lemma\nobreakspace \ref {lemma_caracX}. Moreover, Equation\nobreakspace \textup {(\ref {eq_sigmatotmRIS})} yields an asymptotic rate of entropy production in the small coupling limit.
 
Both behaviours are displayed by qubits interacting via their dipoles. When the rotating wave approximation is taken, the total number operator is preserved. This physical symmetry implies the detailed balance condition and moreover vanishing entropy production in the adiabatic limit. On the other hand, without this approximation we do not find the same symmetry, and recover divergent total entropy production with an asymptotic rate in the adiabatic limit.

\appendix 

\section{Proofs for Section\nobreakspace \ref {sec:Adiabatic} \label{App:Tanaka}}

\paragraph{Proof of Lemma\nobreakspace \ref {lemma_UnifBoundLPn}}
We neglect the subscripts $k$. Since $\L^P$ is simple, we have $\overline{e}^j\,\L = \sum_i \overline{e}^j.e^i \,  P^i + \overline{e}^j\, \L^Q $. Using that projectors $P^i$, $P^j$ associated with different peripheral eigenvalues $e^i$, $e^j$ satisfy $P^i P^j = P^j P^i =0$, we consider the ergodic sum
\begin{equation*}	
\frac{1}{M}\sum_{n=0}^{M-1} (\overline{e}^j\L)^n
= P^j + \frac{1}{M} \sum_{i \neq j} \frac{1- (\overline{e}^j.e^i)^M\!\!}{1- (\overline{e}^j.e^i)^{\hphantom{N}}\!\!}\, P^i  +\frac{1}{M} \sum_{n=0}^{M-1} (\overline{e}^j)^n \, (\L^Q)^n.
\end{equation*}
The left-hand side is a contraction for all $M$. On the right-hand side, the second term has norm going to zero as $M\to\infty$, and the spectral radius formula and $\spr (\L^Q)< 1$ imply that this is true of the third term as well. Therefore, $\|P^j\|\leq  1$ and since $P^j$ is a projector we have equality.

 Next, let us write the finitely many eigenvalues as $e^j=e^{\i x_j2\pi}$, with $x_j\in\R$, $j=1, \dots, r$. By Dirichlet's approximation theorem, for all $q\in\N^*$, there exists integers $\{p_j\}_{j=1,\dots,r}$ and $n_q \geq q$ in~$\N$ such that $|n_qx_j-p_j|\leq 1/q$, for all $j=1, \dots, r$, hence $(e^j)^{n_q}-1=O(1/q)$. Now, $Q\L^n=(\L^Q)^n$ so that 
 $\|Q\L^n\|\leq \ell^n\to 0$ as $n\to \infty$, and $\L^n=\sum_{j}(e^j)^nP^j+Q\L^n$, with $\|P^j\|=1$. Considering an increasing subsequence $(\tilde n_q)_q$ of $(n_q)_q$, we have $\lim_{q\to \infty}\L^{\tilde n_q}= P$. As $\L^{\tilde n_q}$ is a contraction and~$P$ a projector, we get $\|P\|=1$. 

 Finally, if $\L$ is CPTP, we have $\L^{\tilde{n}_q}$ is CPTP for any $\tilde{n}_q$, and thus the limit $P$ and composition $P\L$ are CPTP.  

\paragraph{Proof of Lemma\nobreakspace \ref {lem:eigC2}}
The eigenvalues of $\L^P(s)$ are simple roots of a polynomial whose coefficients are $C^2$ because $\L^P(s)$ is $C^2$ by assumption. But simple roots of a monic polynomial are smooth functions of the coefficients of the polynomial (see~\cite{Marden}), so the eigenvalues of $\L^P(s)$ are locally~$C^2$. The gap assumption \textup{H2} allows us to label them in such a way that these eigenvalues will be $C^2$ on~$[0,1]$.

Next, following~\cite{Kato}, we note that if an operator-valued function $T(s)$ is differentiable and invertible  in a neighbourhood of $s$, then $T\inv(s)$ is differentiable in that same neighbourhood, and $\frac{\d}{\d s} T(s)\inv = - T(s)\inv T'(s) T(s)\inv$. Applying this to $R^P(s,z):=( z-\L^P(s))\inv$, we obtain that $s\mapsto R^P(s,z)$ is twice differentiable on any interval of $[0,1]$ on which $z$ is not an eigenvalue of~$\L^P(s)$. Choose some peripheral eigenvalue $e^j(s)$ and fix $s_0$. From our gap and bound assumptions \textup{H2} and \textup{H4}, there exists a circle $\Gamma$ and $\delta>0$ such that $\Gamma$ encircles $e^j(s)$ for $|s-s_0|<\delta$, but stays a uniform distance away from $e_i(s)$ for any $i\neq j$. Then for any $s$ in the above neighbourhood of $s_0$, the spectral projector onto $e^j(s)$ is equal to
\begin{equation*}	
 P^j(s) =  \frac{1}{2\i\pi} \int_\Gamma R^P(s,z)\d z,
\end{equation*}
and the preceding discussion shows that $P^j(s)$ is $C^2$ on $[0,1]$.

\paragraph{Proof of Proposition\nobreakspace \ref {prop:Qkilled}}
We omit $T$ subscripts for simplicity.
An equivalent expression to~equation\nobreakspace \textup {(\ref {eq:Qkilled})} is
\begin{equation}	\label{eq_LminusLP}
\big\|\big(\L_k \L_{k-1}\dotsm  \L_1  - \L_k^P \L_{k-1}^P\dotsm  \L_1^P\big) P_0 + \big(\L_k \L_{k-1}\dotsm  \L_1  - \L_k^Q \L_{k-1}^Q\dotsm  \L_1^Q\big) Q_0 \big\|\leq  \frac{C^P}{T(1-\ell)}.
\end{equation}
Let us first consider
\begin{align}	
\big(\L_k \L_{k-1}\dotsm  \L_1  - \L_n^P \L_{k-1}^P\dotsm  \L_1^P\big) P_0. \label{eq:P0term}
\end{align}
By  writing $\L_n =  \L^P_n + \L^Q_n$ for each $0\leq n \leq k$, this expression can be expanded as $\L^Q_k \L^Q_{k-1} \ldots \L^Q_1 P_0$ plus a sum of terms of four different forms:
\begin{gather}
\big(\prod_{a\in A_d} \L_a^Q\big)\big(\prod_{b\in B_d} \L_b^P\big)\ldots \big(\prod_{a\in A_1} \L_a^Q\big)\big(\prod_{b\in B_1} \L_b^P\big)\, P_0, \label{eq_factorform1}\\
\big(\prod_{b\in B_{d+1}} \L_b^P\big)\big(\prod_{a\in A_d} \L_a^Q\big)\big(\prod_{b\in B_d} \L_b^P\big)\ldots \big(\prod_{a\in A_1} \L_a^Q\big)\big(\prod_{b\in B_1} \L_b^P\big)\, P_0, \label{eq_factorform2}\\
\big(\prod_{a\in A_{d+1}} \L_a^Q\big)\big(\prod_{b\in B_{d+1}} \L_b^P\big)\ldots \big(\prod_{a\in A_2} \L_a^Q\big)\big(\prod_{b\in B_2} \L_b^P\big) \big(\prod_{a\in A_1} \L_a^Q\big)\, P_0,
\label{eq_factorform3}\\
\big(\prod_{b\in B_{d+1}} \L_b^P\big)\big(\prod_{a\in A_d} \L_a^Q\big)\big(\prod_{b\in B_{d-1}} \L_b^P\big)\ldots \big(\prod_{a\in A_2} \L_a^Q\big)\big(\prod_{b\in B_2} \L_b^P\big) \big(\prod_{a\in A_1} \L_a^Q\big)\, P_0, \label{eq_factorform4}
\end{gather}
where $d\geq 1$, any $A_n$ or $B_n$ is a nonempty set of consecutive elements of $\{1,\ldots, k\}$ such that in every term of one of the above forms, $\bigcup_n A_n \cup \,\bigcup_n B_n = \{1,\ldots, k\}$, $A_m\cap A_n = B_m \cap B_n=\emptyset$ for $m\neq n$, $A_m \cap B_n = \emptyset$ for any $m,n$, and the products are e.g.
\[\prod_{a\in \{a_0+1,\ldots,a_0+t\}}  \L_a^Q := \L^Q_{a_0+t}\ldots  \L^Q_{a_0+1}, \qquad \prod_{b\in \{b_0+1,\ldots,b_0+t\}}  \L_b^P := \L^P_{b_0+t}\ldots  \L^P_{b_0+1}.\]
Remark that this notation enforces the fact that the term $\L_k^P\ldots \L_1^P P_0$ does not appear. 
Our proof relies on a few simple properties. By Lemma\nobreakspace \ref {lem:projC2}, for any $n$ in $\{1,\ldots, T\}$, we have 
\begin{align}\label{cp}
\|P_n-P_{n-1}\|\leq c^P/T.
\end{align} 
Since e.g. $P_n Q_{n-1}= (P_n-P_{n-1}) Q_{n-1}$, and $\|Q_{n-1}\| \leq \|\one - P_{n-1}\| \leq 2$, we have that  
\begin{equation}
\|P_n Q_{n-1}\| \leq c/T, \qquad \|Q_n P_{n-1}\| \leq c/T. \label{eq_BoundPQ}\end{equation}
where $c= 2c^P$. Next, from assumption \textup{H4} and Lemma\nobreakspace \ref {lemma_UnifBoundLPn} we have, denoting by $|A|$ the cardinal of a set $A$,
\begin{equation} \|\prod_{a\in A_n}  \L_a^Q \| \leq \ell^{|A_n|}, \qquad  \|\prod_{b\in B_n}  \L_b^P \| \leq 1. \label{eq_BoundLPLQ}\end{equation}
We therefore have immediately that 
\[\|\L^Q_k \ldots \L^Q_1 P_0\| \leq c\, \ell^k\, T^{-1}\]
and 
\begin{align*}
\|\eqref{eq_factorform1}\| &\leq (c/T)^{2d-1} \, \ell^{\sum_n |A_n|}, &
\|\eqref{eq_factorform2}\| &\leq (c/T)^{2d} \, \ell^{\sum_n |A_n|},\\
\|\eqref{eq_factorform3}\| &\leq (c/T)^{2d+1} \, \ell^{\sum_n |A_n|}, &
\|\eqref{eq_factorform4}\| &\leq (c/T)^{2d} \, \ell^{\sum_n |A_n|}.
\end{align*}
And the proof now follows from simple combinatorics. Let us just consider terms of the form~\eqref{eq_factorform1}. The index $d$ can run from $1$ to $[\frac k2]$. The index $\alpha:= \sum_n |A_n|$ is constrained by $\alpha \geq d$ and $k-\alpha \geq d$. Once $d$ and $\alpha$ are chosen, the exact factor \eqref{eq_factorform1} is determined by the choice of $|A_1|,\ldots,|A_d|$ and $|B_1|,\ldots,|B_d|$, all of which are $\geq 1$ and with the constraints $\sum_n |A_n|=\alpha$ and $\sum_n |B_n|=k-\alpha$. There are respectively ${\alpha -1 \choose d-1}$ and ${k-\alpha-1 \choose d-1}$ such choices, so that, using the norm estimates derived from \eqref{eq_BoundPQ} and \eqref{eq_BoundLPLQ}, the sum of all terms of the form \eqref{eq_factorform1} has norm smaller than 
\begin{align}\nonumber
\sum_{d=1}^{[\frac k2]}\sum_{\alpha=d}^{k-d} &\Big(\frac{ c}T\Big)^{2d-1}  \, \ell^\alpha \, {\alpha -1 \choose d-1}{k-\alpha-1 \choose d-1}\\ 
&= \frac T{c}\, \sum_{\alpha=1}^{k-1} \ell^\alpha \sum_{d=1}^{\inf(\alpha,k-\alpha)}  \Big(\frac{c^2 }{T^2}\Big)^{d}  \, {\alpha -1 \choose d-1}{k-\alpha-1 \choose d-1} \label{eq_l_optimal1}\\
&\leq  \frac Tc \,\sum_{\alpha=1}^{k-1} \ell^\alpha\,  \Big(\sum_{d=1}^\alpha  \Big(\frac{c^2 }{T^2}\Big)^{d/2}  \, {\alpha -1 \choose d-1}\Big) \Big(\sum_{d=1}^{k-\alpha}  \Big(\frac{c^2 }{T^2}\Big)^{d/2}  \, {k-\alpha -1 \choose d-1}\Big) \nonumber \\
&\leq \frac{c }{T} \, \Big(1+  \frac{c }{T}\Big)^{k-2}\, \sum_{\alpha=1}^{k-1} \ell^\alpha 
\leq  \frac{c \,\exp c}{T(1-\ell)}, \label{eq_l_optimal2}
\end{align}
The other forms, \eqref{eq_factorform2}--\eqref{eq_factorform4}, are very similar and yield upper bounds of the same type.
We can similarly expand
\[(\L_k \L_{k-1}\ldots \L_1 -\L^Q_k \L^Q_{k-1}\ldots \L^Q_1)Q_0\]
as $\L^P_k \L^P_{k-1} \ldots \L^P_1 Q_0$, plus a sum of terms similar to (\ref{eq_factorform1}--\ref{eq_factorform4}), with $P$ and $Q$ exchanged. The bounds and enumerations are similar to the ones given above. As $\|\L^P_k \L^P_{k-1} \ldots \L^P_1 Q_0\|\leq c/T$, we get the result.

\paragraph{Proof of Proposition\nobreakspace \ref {prop:general-ad}}
As previously mentioned, we follow the strategy of \cite{Tan11} with our different form of the adiabatic approximation, for $T\geq T_0$, $T_0$ depending on $c^P$ and $N_{\text{max}}$ only. For notational simplicity, we omit the $T$ subscripts, and will say that an expression is $O^P(T^{-\alpha})$ if there exists a universal numerical constant~$C$ depending on $c^P$ and $N_{\mathrm{max}}$ alone, such that the norm of the expression is bounded by $C T^{-\alpha}$. We want to show (see Equations\nobreakspace \textup {(\ref {eq:Kat})} and\nobreakspace  \textup {(\ref {eq:Phase})})
\begin{equation*}
\L^P_k \L^P_{k-1} \dotsm \L^P_1 \L^P_1 P_0 - \Kat{k} \Phase{k} \, P_0 =O^P({T}^{-1}).
\end{equation*}
By Equation\nobreakspace \textup {(\ref {eq_Kat2})}, this is equivalent, for $T\geq T_0$, to
\begin{align}
	\Omega_k := \Phase{k}^\dagger\Kat{k}^\dagger \L^P_k \L^P_{k-1} \dotsm \L^P_1 P_0 &=  P_0 + O^P({T}^{-1}). \label{eq:ultgoal}
\end{align}
Note that $\Omega_0 = P_0$, that $\Omega_k$ is uniformly bounded by equation\nobreakspace \textup {(\ref {eq_unifboundA})}, and that $P_0 \Omega_k = \Omega_k$.
We start by rearranging $\Omega_k$: with $\Theta_k := \Phase{k}^\dagger\Kat{k}^\dagger {\L}^P_k \Kat{k-1}\Phase{k-1}$ for $k \in \N$, we have $\Omega_k = \Theta_k \Omega_{k-1}$ and $\Theta_0 = P_0$, so that 
\begin{equation*}
	\Omega_k = P_0 + \sum_{n=1}^k (\Theta_n - P_{0})\, \Omega_{n-1} = P_0 + \sum_{n=1}^k (V_n - V_{n-1})\, \Omega_{n-1},
\end{equation*}
where we let $V_n := \sum_{m=1}^n (\Theta_m - P_{0})$ and $V_0 := 0$. By summation by parts, we have
\begin{align}
	\Omega_k &= P_0 + V_k \Omega_{k-1} - \sum_{n=1}^{k-1} V_n(\Omega_n - \Omega_{n-1}) = P_0 + V_k \Omega_{k} - \sum_{n=1}^{k-1} V_n(\Theta_n - \Theta_{0})\Omega_{n-1}. \label{eq:goal}
\end{align}
Therefore, if $\Theta_n - \Theta_{0} = O^P({T}^{-1})$ and $V_n = O^P({T}^{-1})$ for each $n \in \{0,\dotsc,k\}$, then we will have Equation\nobreakspace \textup {(\ref {eq:ultgoal})}. This is a consequence of the next two lemmas.
\begin{lemma} \label{lemma_boundThetas}
For any $k\in\{1,\ldots,T\}$, we have
\begin{gather*}
P_0^j (\Theta_k - \Theta_0)P_0^j = O^P({T}^{-2}) \mbox{ for any }j,\\
P_0^j (\Theta_k - \Theta_0)P_0^\ell = O^P({T}^{-1}) \mbox{ for any }j\neq\ell,\\
Q_0 (\Theta_k - \Theta_0) = (\Theta_k - \Theta_0) Q_0 = 0. 
\end{gather*}
\end{lemma}

\begin{proof}
The third relation is obvious. To prove the first, fix $j$. Using relations \eqref{eq:Kat} and \eqref{eq_Kat2} we have
\begin{align*}
	P_0^j (\Theta_k - \Theta_{0}) P_0^j&=P^j_0 (\Phase{k}^\dagger\Kat{k}^\dagger \L^P_k \Kat{k-1}\Phase{k-1}-P_{0})P^j_0 \notag\\
	&= P^j_0 (\Kat{k}^\dagger P_k^j \Kat{k-1}-P_{0})P^j_0 \notag\\
	&= P_0^j \Kat{k}^\dagger P_k^j P_{k-1}^j \Kat{k-1}P_0^j - P_0^j \Kat{k}^\dagger \kat{k} \Kat{k-1} P_0^j \notag\\
	&= P_0^j \Kat{k}^\dagger \Big(P_k^jP_{k-1}^j - P_k^jP_{k-1}^j\,\big(\one-(P_k^j-P_{k-1}^j)^2\big)^{-1/2}\Big) \Kat{k-1} P_0^j 
\end{align*}
and using  $\|P_k^j-P_{k-1}^j\| \leq c^P/T$ and $\big(\one-(P_k^j-P_{k-1}^j)^2\big)^{-1/2} = \one  -\frac{1}{2}(P_k^j-P_{k-1}^j)^2 + O^P((P_k^j-P_{k-1}^j)^4) $ we have the first relation.

We now consider the second relation. Fix therefore $j$ and $\ell$, Note that for $j \neq \ell$,
\begin{equation*}
	P^j_0 (\Theta_k - \Theta_0) P^\ell_0 = Z_{k-1}^{j\ell}R_{k}^{j\ell},
\end{equation*}
where
\begin{equation*}
	Z_{k}^{j\ell} := \prod_{n=1}^k \overline e _n^j. e_n^\ell \quad \mbox{and}\quad R_{k}^{j\ell} = \Kat{k}^\dagger P_k^j P_{k-1}^\ell \Kat{k-1}.
\end{equation*}
We have by Lemma\nobreakspace \ref {lemma_UnifBoundLPn}
\begin{equation}
	\|R_{k}^{j\ell}\|
		\leq \| P_k^j\|\| P_{k-1}^\ell - P_{k}^\ell\|
		= O^P({T}^{-1}) \label{eq:R}
\end{equation}
Therefore $\|P_0^j(\Theta_k - \Theta_0)P_0^{\ell}\|= O^P({T}^{-1})$.
\end{proof}

\begin{lemma} \label{lemma_boundVs}
For any $k\in\{1,\ldots,T\}$, we have
\begin{gather*}
P_0^j\,V_k\,P_0^j = O^P({T}^{-1}) \mbox{ for any }j,\\
P_0^j\,V_k\,P_0^\ell = O^P({T}^{-1}) \mbox{ for any }j\neq\ell,\\
Q_0 \,V_k = V_k\, Q_0 = 0. 
\end{gather*}
\end{lemma}

\begin{proof}
Again the third relation is obvious. The first follows from the first relation in Lemma\nobreakspace \ref {lemma_boundThetas} and the definition of $V_n$. To prove the second relation, fix again $j$ and $\ell$. We have
\begin{equation*}
		P_0^j \, V_k\, P_0^\ell = P_0^j \sum_{n=1}^k (\Phase{n}^\dagger\Kat{n}^\dagger \L^P_n \Kat{n-1}\Phase{n-1} - P_0) P_0^\ell = \sum_{n=1}^k Z_{n-1}^{j\ell} R_{n}^{j\ell}.
\end{equation*}
Note that $Z_{n-1}^{j\ell} = \frac{Z_{n}^{j\ell}-Z_{n-1}^{j\ell}}{\overline e _n^j.e_n^\ell-1}$, so that summation by parts yields
	\begin{align*}
		P_0^j \sum_{n=1}^k (\Theta_n - P_0 ) P_0^\ell &= \sum_{n=1}^k (Z_{n}^{j\ell} - Z_{n-1}^{j\ell})\frac{R_{n}^{j\ell}}{\overline e _n^j.e_n^\ell-1} \\
		&= \frac{Z_{k}^{j\ell}R_{k}^{j\ell}}{\overline e _k^j.e_k^\ell-1} - \frac{Z_{0}^{j\ell}R_{1}^{j\ell}}{\overline e _1^j.e_1^\ell-1} - \sum_{n=1}^{k-1} Z_{n}^{j\ell}\bigg(\frac{R_{n+1}^{j\ell}}{\overline e _{n+1}^j.e_{n+1}^\ell-1}-\frac{R_{n}^{j\ell}}{\overline e _n^j.e_n^\ell-1}\bigg).
	\end{align*}
By Lemma\nobreakspace \ref {lemma_UnifBoundLPn} and our gap assumption H2, the first two (boundary) terms are $O^P({T}^{-1})$. Moreover, the remaining summand is
\begin{equation}
	Z_{n}^{j\ell}\bigg(\frac{R_{n+1}^{j\ell}}{\overline e _{n+1}^j.e_{n+1}^\ell-1}-\frac{R_{n}^{j\ell}}{\overline e _{n}^j.e_{n}^\ell-1}\bigg) = \frac{Z_{n}^{j\ell}}{\overline e _{n}^j.e_{n}^\ell-1}(R_{n+1}^{j\ell}-R_{n}^{j\ell}) + O^P({T}^{-2}) \label{eq:summand}
\end{equation}
by Lemma\nobreakspace \ref {lem:eigC2} which implies that $\bar{e}_{n+1}^j e^\ell_{n+1} = \bar{e}^j_n e^\ell_n + O^P({T}^{-1})$. We have
\begin{align*}
		R_{n+1}^{j\ell}-R_{n}^{j\ell} &= \Kat{n+1}^\dagger P_{n+1}^j P_{n}^\ell \,\Kat{n} - \Kat{n}^\dagger P_{n}^j P_{n-1}^\ell \Kat{n-1} \notag \\
			&= \Kat{n}^\dagger  (\kat{n+1}^\dagger P_{n+1}^jP_n^\ell\,\kat{n}-P_n^j P_{n-1}^\ell) \Kat{n-1} \notag\\
			&= \Kat{n}^\dagger  P_n^j\big( P_{n+1}^j P_n^\ell  -P_n^j P_{n-1}^\ell\big)  P_{n-1}^\ell\Kat{n-1} + O^P({T}^{-2}) \notag\\
			&= \Kat{n}^\dagger  P_{n}^j \big( (P_n^\ell-P_{n+1}^\ell)- (P_{n-1}^\ell-P_n^\ell)\big) P_{n-1}^\ell\Kat{n-1} + O^P({T}^{-2}). \notag
\end{align*}
By Lemma\nobreakspace \ref {lem:projC2} and the Taylor-Lagrange formula, the quantity 
\[  (P_n^\ell-P_{n+1}^\ell)- (P_{n-1}^\ell-P_n^\ell)= \big(P^\ell(\tfrac{n}{T})-P^\ell(\tfrac{n+1}{T})\big)- \big(P^\ell(\tfrac{n-1}{T})-P^\ell(\tfrac{n}{T})\big)\]
is an $O^P({T}^{-2})$. Summation over $n$ gives the second relation in Lemma\nobreakspace \ref {lemma_boundVs}.
\end{proof}

\section{Proofs for Section\nobreakspace \ref {sec:Relative-entropy} \label{App:RelativeEntropy}}

To prove Proposition \ref{eq_perturbentropy} we need two technical results. Assume $\eta$ has spectral decomposition $\eta = \sum_{j} \mu_j p_j $,  and denote by $R_0$ its resolvent, i.e. $R_0(z) = (\eta-z)^{-1}$ for $z\not\in\sp \eta$. 
\begin{lemma} \label{lem:perturb}
Let $\eta$ and $R_0(z)$ be as above and let $D$ and $M$ be two matrices on $\H$. Let $f$ be holomorphic in an open domain $\Omega\subset \C$ such that $\sp\eta\subset \Omega$ and let $\Gamma$ be a positively oriented contour in $\Omega$ encircling $\sp \eta$. Let
\begin{align*}
	T_n(M,D,f) := \tr \Big( -\frac{1}{2\i\pi} M \int_{\Gamma} R_0(\zeta)(D R_0(\zeta))^n f(\zeta) \d\zeta \Big)
\end{align*}
Then,
\begin{align}	
T_1(M,D,f) = - \sum_i \tr(M p_i D p_i) f'(\mu_i) - \sum_{ i< j} \tr\big(M (p_i D p_j + p_j D p_i)\big) \frac{f(\mu_i) - f(\mu_j)}{\mu_i-\mu_j}. \label{eq:2R0full}
\end{align}
and, when $[M,\eta]=0$, 
\begin{align}
T_2(M,D,f) &=   \sum_{i} \tr(M D p_i D p_i) \frac{f''(\mu_i)}{2} \notag \\ 
		& \qquad {}+ \sum_{ i \neq j} \tr(M D p_j D p_i) \frac{f'(\mu_i)}{\mu_i-\mu_j} + \tr(MD p_j D p_i) \frac{f(\mu_j)-f(\mu_i)}{(\mu_i-\mu_j)^2}.
\end{align}

\end{lemma}

\begin{proof}
We write the resolvent $R_0(\zeta) = \sum_i (\mu_i-\zeta)^{-1} p_i$.  We compute
\begin{align*}
	T_1(M,D,f)&= -\tr\Bigg(\frac{1}{2\i \pi} \sum_{ij} \int_\Gamma  M p_i D p_j (\mu_i -\zeta)^{-1}(\mu_j - \zeta)^{-1} f(\zeta) \d \zeta\Bigg). 
\end{align*}
A standard application of Cauchy's integral formula shows
\begin{align*}
	\frac{1}{2\i \pi} \int_\Gamma (\zeta-\mu_i)^{-1}(\zeta-\mu_j)^{-1} f(\zeta) \d\zeta = \Bigg\{ 
		\begin{array}{c l}
			f'(\mu_i) & \text{if } i=j \\
			\tfrac{f(\mu_i)-f(\mu_j)}{\mu_i-\mu_j} & \text{if } i\neq j 
		\end{array}
\end{align*}
and the result follows from the symmetry
\begin{align*}
	\sum_{i \neq j}\tr (Mp_i D p_j) \frac{f(\mu_i)-f(\mu_j)}{\mu_i-\mu_j} = \sum_{i < j}\tr \big(M(p_i D p_j + p_j D p_i)\big) \frac{f(\mu_i)-f(\mu_j)}{\mu_i-\mu_j}.
\end{align*}
Similarly, assuming $[M,\eta]=0$, and using the cyclicity of the trace,
\begin{align*}
	T_2(M,D,f)
		 &= \tr \Big( -\sum_{i,j,k} \frac{1}{2\i \pi} \int_\Gamma M p_i D p_j D p_k \,  (\mu_i-\zeta)^{-1}(\mu_j-\zeta)^{-1}(\mu_k-\zeta)^{-1}f(\zeta) \d\zeta \Big) \\
		 &= \tr \Big( \sum_{i,j} \frac{1}{2\i \pi} \int_\Gamma M D p_j D p_i\, (\zeta-\mu_i)^{-2}(\zeta-\mu_j)^{-1}f(\zeta) \d\zeta \Big). 
\end{align*}
A standard computation shows
\begin{align*}
	\frac{1}{2\i \pi} \int_\Gamma (\zeta-\mu_i)^{-2}(\zeta-\mu_j)^{-1} f(\zeta) \d\zeta = \Bigg\{ 
		\begin{array}{c l}
			\tfrac{f''(\mu_i)}{2} & \text{if } i=j \\
			\tfrac{f(\mu_j)-f(\mu_i)}{(\mu_j-\mu_i)^2} + \tfrac{f'(\mu_i)}{(\mu_i-\mu_j)} & \text{if } i\neq j 
		\end{array}
\end{align*}
and the result follows.
\end{proof}

\begin{corollary} \label{cor:perturb}
In the setup of Lemma\nobreakspace \ref {lem:perturb}, if $[M,\eta]=0$, we have
\begin{gather}
T_1(M,D,f)
	= -\tr\left(M D f'(\eta)\right) \label{eq:2R0commute}, \\
T_2(\one,D,f) = \sum_i \tr\big((D p_i)^2\big) \frac{f''(\mu_i)}{2} + \sum_{ i < j} \tr( D p_j D p_i) \frac{f'(\mu_i)-f'(\mu_j)}{\mu_i-\mu_j}. \label{eq:3R0id}
\end{gather}
\end{corollary}

We can now compute the expansion of the relative entropy of a state with respect to its perturbation up to second order in the perturbation to end the proof of Equation\nobreakspace \textup {(\ref {eq_perturbentropy})}.

With $\epsilon(\eta) := \inf\sp\eta$, let~$\Gamma$ be a (positively oriented) rectangular contour satisfying $\operatorname{dist}(\Gamma, \sp \eta) \geq \tfrac{1}{2}\epsilon(\eta)$ and $\operatorname{dist}(\Gamma, 0) \geq \tfrac{1}{2}\epsilon(\eta)$. For some $D_\eta > 0$ small enough and $\|D_\ell\| \leq D_\eta$, $\Gamma$ encloses the spectrum of $\eta+D_\ell$. Then, for $\|D_\ell\| \leq \tfrac{1}{4}\epsilon(\eta)$ and all $\zeta \in \Gamma$, we have $\|R_0(\zeta) D_\ell \| \leq 1/2$ and the Neumann series
\[
	R_\ell(\zeta)=R_0(\zeta) \big(\one + D_\ell R_0(\zeta)\big)^{-1}= R_0(\zeta) \sum_{n\geq0}(-D_\ell R_0(\zeta))^n,
\]
where $R_\ell$ denotes the resolvent of $\eta + D_\ell$. Therefore using linearity of the trace
\begin{align*}
	S(\eta+D_1|\eta+D_2) &= \tr\Big(-\frac{1}{2\i \pi}\int_\Gamma\zeta\log\zeta R_1(\zeta) \d\zeta + (\eta+D_1) \frac{1}{2\i \pi}\int_\Gamma\log\zeta R_2(\zeta) \d\zeta\Big) \\
	&= \tr\Big( -\frac{1}{2\i \pi}\int_\Gamma\zeta\log\zeta R_0(\zeta)\sum_{n\geq0}(-D_1 R_0(\zeta))^n \d\zeta \\
	&\qquad\qquad+ (\eta+D_1)\, \frac{1}{2\i \pi}\int_\Gamma\log\zeta R_0(\zeta)\sum_{n\geq0}(-D_2 R_0(\zeta))^n \d\zeta\Big)  \\
	&= \tr \Big( \frac{1}{2\i \pi} \int_\Gamma\zeta \log \zeta R_0(\zeta)  D_1 R_0(\zeta) - \zeta \log \zeta R_0(\zeta)(D_1 R_0(\zeta))^2 \d \zeta \Big) -\tr(D_1\log\eta)\\
	&\qquad{}-\tr \Big(
	(\eta + D_1) \frac{1}{2\i \pi}\int_\Gamma \big(\log \zeta R_0(\zeta) D_2 R_0(\zeta) - \log \zeta R_0(\zeta) D_2 R_0(\zeta) D_2 R_0(\zeta)\big)\d \zeta \Big) \\
	&\qquad{} + r(\Gamma, \eta, D_1,D_2)\\
	&= - T_1(\one, D_1, f) + T_2(\one, D_1, f) - \tr(D_1\log\eta) + T_1(\eta,D_2,g) + T_1(D_1,D_2, g) \\
	&\qquad{} \qquad{}  - T_2(\eta,D_2,g)-T_2(D_1,D_2,g)  + r(\Gamma, \eta, D_1,D_2), 
\end{align*}
where
\begin{align*}
r(\Gamma, \eta,D_1,D_2) :&{\!}= \tr\Big(-\frac1{2\i\pi} \int_{\Gamma} \zeta \log \zeta \, R_0(\zeta)\, \sum_{n\geq 3} \big(-D_1 R_0(\zeta)\big)^n\, \d \zeta\Big)\\
 &\qquad + \tr\Big(-\frac{\eta+D_1}{2\i\pi} \int_{\Gamma} \zeta \log \zeta \, R_0(\zeta)\, \sum_{n\geq 3} \big(-D_2 R_0(\zeta)\big)^n\, \d \zeta\Big).
\end{align*}
Then, using and $\|D_\ell R_0(\zeta)\|^3 \leq \tfrac{8}{\epsilon(\eta)^3} \|D_\ell\|^3$ and $|\log \zeta| \leq c \log(\tfrac{1}{2}\epsilon(\eta))$ for all $\zeta \in \Gamma$, we have the estimate
\begin{align*}
	|r(\Gamma, \eta,D_1,D_2)| &\leq c \dim\H \frac{|\log(\epsilon(\eta))| + \log 2}{\epsilon(\eta)^4} (\|D_1\| + \|D_2\|)^3,
\end{align*}
where $c$ is a numerical constant (possibly changing form estimate to estimate). Similarly, we get that $|T_2(D_1,D_2,g)|$ is bounded above by the same quantity, which defines $C_\eta$.

Finally, Lemma\nobreakspace \ref {lem:perturb} and Corollary\nobreakspace \ref {cor:perturb} yield
\begin{align*}
 T_1(\one,D_1,f) &= -\tr\big(D_1(\log \eta + \one)\big) = -\tr(D_1 \log \eta),\\
 T_1(\eta,D_2,g) &= -\tr(\eta D_2\eta\inv) = -\tr(D_2) =0,\\
 T_1(D_1,D_2,g) &= -\sum_i \tr (D_1 p_i D_2 p_i) \frac{1}{\mu_i} - \sum_{i< j}\tr\big(D_1(p_iD_2p_j + p_j D_2 p_i)\big) \frac{\log(\mu_i)-\log(\mu_j)}{\mu_i-\mu_j},\\
T_2(\one,D_1,f) &= \sum_i \tr\big((D_1p_i)^2\big) \frac{1}{2 \mu_i} + \sum_{i<j} \tr(D_1p_j D_1p_i) \frac{\log(\mu_i)-\log(\mu_j)}{\mu_i-\mu_j},\\
T_2(\eta,D_2,g) &= -\sum_{i}  \tr((D_2 p_i)^2)\frac{1}{2 \mu_i} +\sum_{i < j}  \tr(D_2 p_j D_2 p_i)\frac{\log(\mu_j)-\log(\mu_i)}{\mu_i-\mu_j},
\end{align*}
and putting things together,
\begin{align}  \nonumber
S(\eta+D_1|\eta+D_2) &= \sum_i\tr \Big(\big((D_1-D_2)p_i\big)^2\Big)(2\mu_i)\inv \label{eq_Sperturbed}\\
	&\qquad{} + \sum_{i<j} \tr((D_1 -D_2)p_j (D_1 - D_2)p_i) \frac{\log(\mu_i) -\log(\mu_j)}{\mu_i-\mu_j}  + O_\eta(\|D\|^3),
\end{align}
where $O_\eta(\|D\|^3)$ denotes a term that is bounded in absolute value by $C_\eta (\|D_1\|+\|D_2\|)^3$ for $\|D_1\|$ and $\|D_2\|$ small enough.

\section{Proofs for Section\nobreakspace \ref {sec:Application} \label{App:Application}}

\paragraph{Proof of Proposition\nobreakspace \ref {prop:normTospr}}
By the discussion in Section\nobreakspace \ref {ssec:RISmathdesc}, every CPTP map is a contraction, and therefore $s\mapsto \L(s)$ satisfies \textup{H1}. By the Perron-Frobenius theorem, peripheral eigenvalues are simple. In addition, they must be a finite subgroup of $S^1$; since their number is bounded by $N_{\max}$ as defined in \eqref{eq_unifboundNs}, there is a gap between them, so, using continuity, \textup{H2} is satisfied. 

Then there exists $z(s)$ such that $\sp\,\L(s) \cap S^1= S_{z(s)}$. The gap between peripheral eigenvalues and \textup{wH4} (and in particular, \textup{H4}) imply that for any $e\in S_{z(s)}$, there exists a neighbourhood $O$ of $e$ such that for any $s'\in[0,1]$, the only possible eigenvalue of $\L(s')$ in $O$ is $e$. Therefore, if $e$ is a peripheral eigenvalue of $\L(s)$ for some $s$, then the results of \S 5, chapter 2 of \cite{Kato} imply that $e$ is still an eigenvalue of $\L(s')$ for $s'$ in a neighbourhood of $s$. A connectedness argument then shows that the peripheral spectrum of $\L(s)$ does not depend on $s$, and it is then necessarily of the form~$S_z$. 

Assume first \textup{H4}. Because peripheral eigenvalues are simple eigenvalues of $\L(s)$, which is a $C^2$ function of $s$, the same argument as for Lemma\nobreakspace \ref {lem:projC2} shows \textup{H3}, and the proof is complete. If we now assume \textup{wH4} in place of \textup{H4}, then the proof of Lemma\nobreakspace \ref {lemma_wH4H4} shows that there exists $m_0$ such that for $m\geq m_0$, the map $s\mapsto \L^m(s)$ satisfies \textup{H1}, \textup{H4}, and the peripheral part $(\L^{m})^P(s)$ is continuous. Again the proof of Lemma\nobreakspace \ref {lem:projC2} shows that $s\mapsto \L^m(s)$ satisfies \textup{H3}. Therefore, we only need to make sure that we choose $m$ so that the peripheral eigenvalues of every $\L^m(s)$ are still simple. Since the peripheral eigenvalues are the set $S_z$ for some $z$, it is enough to consider $m\geq m_0$ such that $\gcd(m,z)=1$.

\section{Proofs for Section\nobreakspace \ref {sec:scl} \label{App:scl}}

\paragraph{Proof of Lemma\nobreakspace \ref {lem:expandUwU}}

Let us drop the subscripts $k$ and $T$ for notational simplicity. Since $U = \exp( -\i\tau (\ham_0 + \lambda v))$, this $U$ is analytic in $\lambda$. Since we assume $\omega := \sysstate\invar \otimes \envstate$ admits a second order expansion  in $\lambda$ as well, we have
\begin{align*}	
\omega &= \omega^\lambda = \omega^{(0)} + \lambda \omega^{(1)} + O(\lambda^2), &
U &= U^{\lambda} = U^{(0)} + \lambda U^{(1)} + O(\lambda^2).
\end{align*}

Set $R^\lambda(z) = (\ham_0 + \lambda v - z)\inv$ to be the resolvent of the coupled Hamiltonian and $R_0(z) = (\ham_0 - z)\inv$. The holomorphic functional calculus yields
\begin{align*}
	U^\lambda \omega^\lambda (U^\lambda)^* &= \frac{1}{(2\i \pi)^2} \int_\Gamma \int_{\Gamma'} \exp(-\i\tau(\zeta-\zeta'))R^\lambda(\zeta)\omega^\lambda R^\lambda(\zeta') \d\zeta' \d\zeta,
	\end{align*}
	where $\Gamma'$ is a contour contained in the interior of $\Gamma$, and both contain the spectrum of the coupled Hamiltonian $h$. We substitute the Neumann expansion $R^\lambda(z) = R_0(z) [ \one + \lambda v R_0(z)]\inv = R_0(z) [ \one - \lambda v R_0(z)  + O(\lambda^2)]$ for $z\in\Gamma$, to obtain
	\begin{align*}
	U^\lambda \omega^\lambda (U^\lambda)^* &= \frac{1}{(2\i \pi)^2} \int_\Gamma \int_{\Gamma'} \exp(-\i\tau(\zeta-\zeta'))R_0(\zeta)[\one - \lambda v R_0(\zeta)]\omega^\lambda R_0(\zeta')[\one - \lambda v R_0(\zeta')] \d\zeta' \d\zeta + O(\lambda^2),
	\end{align*}
	which we rearrange as, using $U^{(0)} \omega^{(0)} (U^{(0)})^* = \omega^{(0)}$,
	\begin{align*}
	\omega^{(0)} + \lambda U^{(0)} \omega^{(1)} (U^{(0)})^* - \frac{\lambda}{(2\i \pi)^2} \int_\Gamma \int_{\Gamma'} \exp(-\i\tau(\zeta-\zeta'))R_0(\zeta)[v R_0(\zeta) \omega^{(0)} + \omega^{(0)} R_0(\zeta') v ] R_0(\zeta') \d\zeta' \d\zeta + O(\lambda^2).
\end{align*}
We  compute these integrals using standard techniques. For example, the first term is
\begin{align*}
	I :=\frac{1}{(2\i \pi)^2} \int_\Gamma \int_{\Gamma'} \exp(-\i\tau(\zeta-\zeta'))R_0(\zeta)v R_0(\zeta) R_0(\zeta') \d\zeta' \d\zeta.
	\end{align*}
	We apply the first resolvent identity on the last factor $R_0(\zeta) R_0(\zeta')$, then perform the $\zeta'$ integral. Next, we write remaining resolvents using the spectral representation $\ham_0 = \sum_i \pi_i E_i$, and use Cauchy's integral formula to obtain
	 \begin{align*}
 I =- \sum_{i} \pi_i v  \pi_i (-\i\tau) - \sum_{i \neq j} \pi_i v  \pi_j \left( \frac{\exp(-\i\tau(E_i-E_j))}{E_i-E_j}-\frac{1}{E_i-E_j}\right).
\end{align*}
We deal with the other term in the same way to obtain, using $[\omega^{(0)}, R_0(z)]=0$,
\begin{align*}
U^\lambda \omega^\lambda (U^\lambda)^* &= 
	\omega^{(0)} + \lambda U^{(0)} \omega^{(1)} (U^{(0)})^* \\
	&\qquad- \lambda \Big[\omega^{(0)}, \sum_{i} \pi_i v  \pi_i (-\i\tau) + \sum_{i \neq j} \pi_i v  \pi_j \Big( \frac{\exp(-\i\tau(E_i-E_j))-1}{E_i-E_j}\Big)\Big] + O(\lambda^2).
\end{align*}

\paragraph{Proof of Lemma\nobreakspace \ref {expp1} }
 We drop the variable $s$ in the notation.
The map $\lambda\mapsto \L^\lambda$ is analytic, and $\lambda=0$ is an (isolated) exceptional point in the sense of \cite{Kato} for the degenerate  eigenvalue 1. For $\lambda\neq 0$ in $D_0\subset \C$,  a sufficiently small neighbourhood of zero independent of $s\in[0,1]$ (see Lemma \ref{unifpert}), the one dimensional projector $\tilde P_1^1(\lambda)$ onto the invariant 
subspace of $ \L^\lambda$ admits a Puiseux series expansion that is single valued, since the constant eigenvalue $1$ is single valued. Moreover, the equality $\|\tilde P_1^1(\lambda)\|=1$ established in Lemma\nobreakspace \ref {lemma_UnifBoundLPn} shows that $\lambda\mapsto \tilde P_1^1(\lambda)$ is actually analytic in~$D_0$. Let $\lambda_0>0$ be small enough so that $\tilde \sysstate(\lambda)=\tilde P_1^1(\lambda)\sysstate\invar(\lambda_0)\neq 0$ and is analytic for $\lambda\in D_0$, possibly shrinking $D_0$. Thanks to the fact that $\L^\lambda$ is CPTP, $\tilde \sysstate(\lambda)\geq 0$ for $\lambda$ real, so that
$\sysstate\invar(\lambda)=\tilde \sysstate(\lambda)/\tr (\tilde \sysstate(\lambda))$ is the unique invariant state of $\L^\lambda$ for $\lambda\in \R$ and admits an analytic extension for $\lambda\in D_0$.

\paragraph{Proof of Lemma\nobreakspace \ref {regproj1}}	
We drop the indices $j$ in the proof. We know that $P(\lambda,s )$ is analytic in $\lambda\in D_0$, $D_0$ in a fixed neighbourhood of the origin. Moreover, $\L^\lambda(s)$ being $C^2$ in $(\lambda,s)\in D_0\times [0,1]$, $P(\lambda,s)$ is~$C^2$ in $(\lambda,s)\in (D_0\setminus\{0\})\times [0,1]$, thanks to its expression as a Riesz integral. Let $\Gamma_r$ be a circle centred at the origin of radius $r$, such that $\Gamma_r\subset D_0$. We have for all $s\in [0,1]$, $\lambda\in D_0\setminus\{0\}$, 
\begin{align}\label{pexplam}
P(\lambda,s )=\sum_{j\geq 0}p_j(s)\lambda^j, \quad \text{with} \quad p_j(s)=\frac{1}{2\i \pi} \int_{\Gamma_r} \frac{P(\lambda,s )}{\lambda^{j+1}}\d\lambda,\ \ j\in \N,
\end{align}
and we can take up to two derivatives with respect to $s$ under the integral sum.
As the derivatives $\partial_s^k P(\lambda,s)$, $k=0,1,2$, are bounded on the compact set $\Gamma_r\times [0,1]$ there exists a $K>0$ such that 
\[
\| p_j^{(k)}(s) \|\leq \sup_{\lambda\in\Gamma_r \atop s\in [0,1]}\|\partial_s^k P(\lambda,s)\|/r^j\leq K/r^j \quad \text{for} \quad j\in \N,\, k\in \{0,1,2\}.
\]
Thus, for $0<|\lambda| <r$, $\partial_s^k P(\lambda,s )$, $k=1,2$ are obtained by taking derivatives in the series (\ref{pexplam}), and the convergence of the resulting series is normal. We have for any $s\in [0,1]$, $0<|\lambda| <r/2$
\[
\Big\|\partial_s^kP(\lambda,s)\Big\|\leq \sum_{j\geq 0}\|p_j^{(k)}(s)\| |\lambda|^j\leq \frac{K}{1-|\lambda|/r}\leq 2K . \qedhere
\]

\paragraph{Proof of Lemma\nobreakspace \ref {lemessy}}

For $\lambda\in \R^*$, we have the spectral decomposition
\[
{\L^\lambda(s)}^mQ(\lambda,s)=\sum_{j}{\tilde e^j(\lambda, s)}^m\tilde P^j(\lambda,s)+\sum_{j>1}{\tilde e_1^j(\lambda, s)}^m(\tilde P_1^j(\lambda,s)+\tilde N_1^j(\lambda,s)),
\]
where $\tilde N_1^j(\lambda,s)$ are the  eigennilpotents associated with degenerate eigenvalues. Hence
\begin{align}\label{upbo}
\| {\L^\lambda(s)}^mQ(\lambda,s)\| & \leq \big(\sup_{t\in[0,1]}\spr \big({\L^\lambda(t)}^mQ(\lambda,t)\big)\big)\Big(\sum_j\|\tilde P^j(\lambda,s)\|+\sum_{j>1}(\|\tilde P_1^j(\lambda,s)\|+\|\tilde N_1^j(\lambda,s)\|)\Big)\nonumber \\
& \leq  \sup_{t\in[0,1]}\big(\spr (\L^\lambda(t)Q(\lambda,t))\big)^m \tilde Q(\lambda)=S(\lambda)^m\tilde Q(\lambda),
\end{align}
where $1<\tilde Q(\lambda):=\sup_{s\in[0,1]}\sum_j\|\tilde P^j(\lambda,s)\|+\sum_j(\|\tilde P_1^j(\lambda,s)\|+\|\tilde N_1^j(\lambda,s)\|)$ and 
\[
S(\lambda):=\sup_{t\in[0,1]}\spr (\L^\lambda(t)Q(\lambda,t))<1.
\]
The eigenvalues $\tilde e_j(\lambda, s)$ are analytic at $\lambda=0$, and $\tilde e_1^j(\lambda, s)$ are given by a converging Puiseux series of the form 
\begin{align}\label{puiseux}
\tilde e_1^j(\lambda, s)-1=\sum_{k=1}^\infty \lambda^{k/p}\alpha_j(k), \quad \text{where} \quad p\in\{1,\dots,\dim \H_\sys\}.
\end{align}
The moduli of these eigenvalues is strictly inferior to $1$ for $\lambda\in\R^*$ small, and we get (\ref{estspr}), with~$r>0$ the exponent corresponding to the largest non-zero leading term in the expansion (\ref{puiseux}). Correspondingly, the eigenprojectors and eigennilpotents admit Puiseux series expansions of the form (\ref{puiseux}), with finitely many negative powers of $\lambda^{1/p}$. Therefore, $\tilde Q(\lambda)$ 
is bounded above, for~$|\lambda|$ small, by a constant times $1/\lambda^{r'}$, with $r'\geq 0$.

Hence, the left hand side of (\ref{upbo}) is bounded above by $1-G(\lambda)$ for any  $0<G(\lambda)\leq1$, if 
\[
m\geq \frac{\log \big(\tilde Q(\lambda)/(1-G(\lambda))\big)}{|\log S(\lambda)|}.
\]
This will be true if we choose $m(\lambda)$ as stated, for a suitable $M_0>0$, when $|\lambda| \neq 0$ is small enough.

In case $\dim \H_\sys =2$, all spectral data are analytic around $\lambda=0$. For the simple eigenvalues $\tilde e^1(\lambda,s),$ $\tilde e^2(\lambda,s)=\overline{\tilde e^1(\lambda,s)}$ that satisfy $|\tilde e^j(\lambda,s)|<1$ for $\lambda\neq 0$ real, one gets, generically,  $|\tilde e^j(\lambda,s)|=1-f^j(s)\lambda^2+O(\lambda^3)$, for some $f_j(s)>0$. The same is true for the real analytic eigenvalues $\tilde e_1^2(\lambda,s)<1$, so that generically, $|\tilde e_1^j(\lambda,s)|=1-f_1^j(s)\lambda^2+O(\lambda^3)$ as well. This yields the exponent $r=2$. The eigenprojections being analytic, and the eigennilponent being identically  zero, the exponent $r'$ equals $0$, which means the lower bound for $m(\lambda)$ contains no logarithmic factor.

{\small
\bibliographystyle{abbrv}
\bibliography{risref}
}

\end{document}